\documentclass[11pt]{article}

\usepackage{amsmath} 
\usepackage{amssymb}
\usepackage{amsthm}
\usepackage{lineno}
\usepackage[ansinew]{inputenc}
\usepackage[T1]{fontenc}
\usepackage{fullpage}
\usepackage{tikz}
\usetikzlibrary{shapes,snakes,decorations}
\usetikzlibrary{arrows.meta}

\usepackage{comment}

\newtheorem{theorem}{Theorem}[section]
\newtheorem{definition}[theorem]{Definition}
\newtheorem{lemma}[theorem]{Lemma}
\newtheorem{corollary}[theorem]{Corollary}
\newtheorem{proposition}[theorem]{Proposition}
\newtheorem{claim}[theorem]{Claim}
\newtheorem{example}[theorem]{Example}

\def\nat{{\mathbb N}}
 \def\real{{\mathbb R}}
 \def\rat{{\mathbb Q}}

\newcommand{\Fix}{{\mathrm{Fix}}}

\newcommand{\coNP}{\ensuremath{\mathsf{coNP}}}
\newcommand{\NP}{\ensuremath{\mathsf{NP}}}
\newcommand{\coUP}{\ensuremath{\mathsf{coUP}}}
\newcommand{\UP}{\ensuremath{\mathsf{UP}}}
\newcommand{\Ptime}{\ensuremath{\mathsf{P}}}
\newcommand{\FP}{\ensuremath{\mathsf{FP}}}
\newcommand{\PPAD}{\ensuremath{\mathsf{PPAD}}}
\newcommand{\PPADS}{\ensuremath{\mathsf{PPADS}}}
\newcommand{\PPA}{\ensuremath{\mathsf{PPA}}}
\newcommand{\PPP}{\ensuremath{\mathsf{PPP}}}
\newcommand{\CLS}{\ensuremath{\mathsf{CLS}}}
\newcommand{\EOPL}{\ensuremath{\mathsf{EOPL}}}
\newcommand{\UniqueEOPL}{\ensuremath{\mathsf{UEOPL}}}

\newcommand{\PLS}{\ensuremath{\mathsf{PLS}}}
\newcommand{\Val}{\ensuremath{\mathsf{Val}}}

\newcommand{\TFNP}{\ensuremath{\mathsf{TFNP}}}
\newcommand{\Tarski}{\ensuremath{\mathtt{Tarski}}}
\newcommand{\Arrival}{\ensuremath{\mathtt{Arrival}}}
\newcommand{\RecursiveArrival}{\ensuremath{\mathtt{RecursiveArrival}}}

\providecommand{\E}{\ensuremath{\mathbb{E}}}

\begin{document}

\title{Tarski's Theorem, Supermodular Games, and the Complexity of Equilibria\thanks{A short conference version of this paper appeared
in the ITCS'2020 conference proceedings (\cite{EPRY-itcs20}). \\ Research partially supported by NSF grants: CCF-2107187, CCF-2212233, CCF-2332922,  CCF-1703295, CCF-1763970, and 2229929.}}

\author{Kousha Etessami\\U. of Edinburgh
\\{\tt kousha@inf.ed.ac.uk}
\and 
Christos H. Papadimitriou\\Columbia U.
\\{\tt christos@cs.columbia.edu} 
\and Aviad Rubinstein\\Stanford U.
\\{\tt aviad@cs.stanford.edu}
\and
Mihalis Yannakakis\\Columbia U.
\\{\tt mihalis@cs.columbia.edu}
}

\date{}

\maketitle

\begin{abstract}
The use of monotonicity and Tarski's theorem in existence proofs of
equilibria is very widespread in economics, while Tarski's theorem is
also often used for similar purposes in the context of verification.
However, there has been relatively little in the way of analysis of
the complexity of finding the fixed points and equilibria guaranteed
by this result.  We study a computational formalism based on monotone
functions on the $d$-dimensional grid with sides of length $N$, and
their fixed points, as well as the closely connected subject of
supermodular games and their equilibria.  It is known that finding
some (any) fixed point of a monotone function can be done in time
$\log^d N$, and we show it requires at least $\log^2 N$ function
evaluations already on the 2-dimensional grid, even for randomized
algorithms.  We show that the general $\Tarski$ problem of finding some
fixed point, when the monotone function is given succinctly (by a
boolean circuit), is in the class $\PLS$ of problems solvable by local
search and, rather surprisingly, also in the class $\PPAD$.  Finding
the greatest or least fixed point guaranteed by Tarski's theorem,
however, requires $d\cdot N$ steps, and is \NP-hard in the white box
model.  For supermodular games, we show that finding an equilibrium is essentially computationally equivalent to the Tarski
problem, and finding the maximum or minimum equilibrium is similarly
harder.  Interestingly, two-player supermodular games where the
strategy space of one player is one-dimensional can be solved in
$O(\log N)$ steps.  We also show that computing (approximating) the
value of Condon's (Shapley's) stochastic games reduces to the $\Tarski$
problem.
An important open problem highlighted by this work is
proving better upper or lower bounds on the
(blackbox query) complexity of the Tarski problem.
\vspace*{0.2in}

\noindent {\bf Keywords.}    Monotone functions, complete lattices,
Tarski's fixed point theorem, computational complexity, randomized/deterministic query complexity, total search
problems, \PPAD{}, \PLS{},
supermodular games, games with strategic complementarities, Nash equilibria,  stochastic games.

\end{abstract}

\setcounter{page}{0}
\thispagestyle{empty}

\newpage

\section{Introduction}

Equilibria are paramount in economics, because guaranteeing their
existence in a particular strategic or market-like framework enables
one to consider ``What happens at equilibrium?'' without further
analysis.  Equilibrium existence theorems are nontrivial to prove.
The best known example is Nash's theorem \cite{Nash51}, whose proof in
1950, based on Brouwer's fixed point theorem, transformed game theory,
and inspired the Arrow-Debreu price equilibrium results \cite{AD54},
among many others.  Decades later, complexity analysis of these
theorems and corresponding solution concepts by computer scientists
has created a fertile and powerful field of research \cite{AGT-book}. 

Not all equilibrium theorems in economics, however, rely on Brouwer's
fixed point theorem for their proof (even though, in a specific sense
made clear and proved in this paper, they could have...).  Many of
the exceptions 
ultimately rely on {\em Tarski's fixed point theorem}
\cite{Tarski55}, stating that all monotone functions on a complete
lattice have a fixed point --- and in fact a whole sublattice of fixed
points with a largest and smallest element
\cite{Topkis79,MR90,Topkis98}.  In contrast to the equilibrium
theorems whose proof relies on Brouwer's fixed point theorem, there
has been relatively little complexity analysis of Tarski's fixed point theorem
and the equilibrium results it enables. (We discuss prior related work at the
end of this introduction.)

Here we present several results in this direction. 
Let $[N] = \{1,\ldots,N\}$.
To formulate the
basic problem, we consider a monotone function $f$ on the
$d$-dimensional grid $[N]^d$, that is, a function
$f:[N]^d\mapsto[N]^d$ such that for all $x, y \in [N]^d$, $x \geq y$ implies $f(x)\geq f(y)$, where ``$x \geq y$'' denotes the standard coordinate-wise partial order on vectors in $\nat^d$, meaning $x \geq y$ if and only if $x_i \geq y_i$ for all $i \in [d]$. 
In the black-box oracle model, we can query this function with specific vectors $x \in [N]^d$.
In the white-box model,
we
assume that the function is presented by a boolean
circuit.
Thus, $d$ and $N$ are the basic parameters to our
model; it is useful to think of $d$ as the dimensionality of the
problem, while $N$ is something akin to the inverse of the desired
approximation $\epsilon$.\footnote{See Proposition \ref{con-disc} for
a formalization of this relationship between $N$ and $\epsilon$.} 

\begin{itemize}
\item Tarski's theorem in the grid framework is easy to prove.
 Let $\bar{1} = (1,\ldots,1)$ denote the ($d$-dimensional) all-1 vector.
  Consider the sequence of grid points $\bar{1}, f(\bar{1}), 
f(f(\bar{1})),
  \ldots, f^i(\bar{1}), \ldots$. From monotonicity of $f$,
by induction on $i$ we get, for all $i \geq 0$,  $f^i(\bar{1}) \leq f^{i+1}(\bar{1})$.
  Unless a fixed point is arrived at, the sum of the coordinates
  must increase at each iteration.  Therefore, after at most $dN$ iterations of $f$ applied to $\bar{1}$, a
  fixed point is found.  In other words $f^{dN}(\bar{1}) = f^{dN+1}(\bar{1})$.

\item This immediately suggests an $O(dN)$ algorithm.  But an
  $O(\log^d N)$ algorithm is also known\footnote{This algorithm
   appears to have been first observed in \cite{DangQiYe}.}:
   Consider the $(d-1)$-dimensional function obtained 
   by fixing the ``input value'' in the $d$'th coordinate
   of the function $f$ with some value $r_d$ (initialize $r_d := \lceil N/2 \rceil$).
   Find a fixed point $z^* \in [N]^{d-1}$ of this $(d-1)$-dimensional 
monotone function $f(z,r_d)$ (recursively).
   If the $d$th coordinate $f_d(z^*,r_d)$ of $f(z^*,r_d)$, is equal to $r_d$, then $(z^*,r_d)$
   is a fixed point of the overall function $f$, and we are done.  Otherwise, a
   binary search on the $d$'th coordinate is enabled:   
   we need to look for a
  larger (smaller) value of $r_d$  if
  $f_d(z^*,r_d) > r_d$  (respectively, if $f_d(z^*,r_d) < r_d$).
  By an easy induction, this establishes the $O(\log^d N)$ upper bound  (\cite{DangQiYe}).

\item
  In Theorem \ref{thm:main-2d} we prove that this
  algorithm is essentially optimal in the black box sense for the case $d=2$.
  We provide a class of monotone functions $f:[N]^2 \rightarrow [N]^2$
  that we call {\em
herringbones}: two monotonic paths, one starting from
  $\bar{1}$ and the other from $\bar{N}$, meeting at the fixed point,
  while all other points in the $ N \times N$ grid are mapped
  diagonally: $f(x)= x+(-1,+1)$ or $x+(+1,-1)$, whichever of the
  points is closer to the monotonic path that contains the fixed
  point.  We prove that any randomized algorithm needs to make
  $\Omega(\log^2 N)$ queries (with high probability) to find the fixed
  point on such herringbone functions.

\item Tarski's theorem further asserts that there is a greatest and a
  least fixed point, and these fixed points are especially useful in
  the economic applications of the result (see for example 
  \cite{MR90}).  It is not hard to see, however, that finding these
  fixed points is \NP-hard, and takes $\Omega(dN)$ time in the black
  box model (see Proposition \ref{lfp:hard}).

\item In terms of complexity classes, the problem $\Tarski{}$ is
  obviously in the class \TFNP{} of total function (total search) problems.  But
  where exactly?  We show (Theorem \ref{thm:tarski-in-pls}) that it belongs in the class
  \PLS{} of local optimum search problems.

\item Surprisingly, $\Tarski{}$ is also in the class 
   $\FP^\PPAD$ of search problems polynomial-time Turing reducible to a Brouwer fixed point problem 
  (Theorem \ref{thm:tarski-in-P-to-PPAD}), and thus,
 by the known fact (\cite{BussJohnson12}) that the class $\PPAD$ is closed 
under P-time Turing reductions 
 $\Tarski{}$ is in $\PPAD$ (Corollary \ref{cor:tarski-in-PPAD}).
This result
presents a heretofore unsuspected connection between two main
sources of equilibrium results in economics.

\item {\em Supermodular games} \cite{Topkis79,MR90,Topkis98} --- or games
with strategic complementarities 
  --- comprise a large and important class of economic models,  with 
  complete lattices as strategy spaces, 
in which a player's best response is a
  monotone function (or monotone correspondence) of the other player's strategies.  They always
  have pure Nash equilibria due to Tarski's theorem.  We show that finding
  an equilibrium for a supermodular game with (discrete) Euclidean grid strategy 
spaces, under a succinct representation of a best response function, is  
essentially computationally equivalent to the problem of finding a Tarski fixed point
of a monotone map 
(Proposition \ref{prop:supmod-to-tarski-reduct} and Theorem \ref{mon-to-sup}).  
If there are two
  players and one of them has a one-dimensional strategy space, we show that a Nash
  equilibrium can be found in logarithmic time (in the size of the strategy spaces).

\item {\em Stochastic games}  \cite{Shapley53, Condon92}.   We show that the 
problems of computing the (irrational) value
of Shapley's discounted stochastic games  to desired accuracy, and 
computing the exact value of Condon's simple stochastic games (SSG), 
are both  P-time reducible to the 
$\Tarski{}$ problem.  The proofs employ known characterizations of the value of both Shapley's stochastic games and Condon's SSGs in terms of 
monotone fixed point equations, which can also be viewed as 
monotone ``polynomially contracting'' maps with a unique fixed
point, and from properties of polynomially contracting maps, see \cite{EY2010}.

\item In Section \ref{sec:gen-lattice} we extend the study of the black-box Tarski fixed point problem
from the setting of monotone functions on the Euclidean grid
$[N]^d$ to a
more general black box model for monotone functions $f:L \rightarrow
L$, over an arbitrary finite lattice $(L,\preceq)$, where the lattice's
elements $L \subseteq \{0,1\}^n$ are encoded as binary strings of some
given length $n$, and where we assume the entire lattice $(L,\preceq)$
is known {\em explicitly} by the algorithm making the queries, 
which moreover has unbounded
computational power, but only has oracle access to the monotone
function $f:L \rightarrow L$.   We show that it is possible
to carry over query upper bounds obtained for the grid lattice $[N]^d$
to this more general black box setting.
In particular, 
generalizing the $\log^d N$
algorithm for Euclidean grids,
we show that in this black box model there is a
deterministic algorithm that computes a fixed point of 
$f:L \rightarrow L$ using $O(\log^d(|L|))$ queries to the function $f$, where 
$d$ is the {\em dimension} of
the lattice $(L,\preceq)$.  The {\em dimension} of a lattice
$(L,\preceq)$, and more generally the dimension of any partial order, 
can be defined as the smallest integer $d \geq 1$ such that the 
relation $\preceq$ is the
intersection of $d$ total orders on the same underlying set $L$.
Equivalently, it is the smallest $d \geq 1$ such that there
is an injective embedding  of $(L,\preceq)$ in the Euclidean grid
$([|L|]^d,\leq)$, where $\leq$ is the standard coordinate-wise partial
order on $[|L|]^d$.  Note that a lower bound of
$\Omega(\log^2 (|L|))$ queries for computing a fixed point of a monotone
function $f:L \rightarrow L$ in this black box model follows directly from our
lower bound of $\Omega(\log^2 N)$ for monotone functions on the
2D grid $f:[N]^2 \rightarrow [N]^2$.  
At present, we do not know any
better lower bound than $\Omega(\log^2 (|L|))$ in this black box model for 
arbitrary finite
lattices.\footnote{Note that this black box model for monotone functions on general finite lattices is very different
from the one considered in \cite{CLT08}, where the lattice itself
is not known explicitly by the algorithm, but 
is only accessible via an oracle for its partial order.
Hence the linear $\Omega(|L|)$ lower bound on the number of queries 
(including queries to the partial order itself) given 
in \cite{CLT08} for finding a fixed point
has no
bearing on the black box model we consider, where the lattice itself is explicitly known,
and only the monotone function is given by an oracle.}

\item  Finally, in Appendix \ref{appendix:B}, we point out
that another intriguing problem called $\Arrival$  (\cite{DGKMW17}),
a reachability problem on directed graphs
whose complexity status remains open,
is reducible to $\Tarski$. This follows directly from a more recent
result by G\"{a}rtner et. al. \cite{GHH21}, who exploited
Tarski fixed point computation to provide a subexponential time
algorithm for $\Arrival$.
\end{itemize}

\noindent {\em Prior related work:} in recent years a number of technical
reports and papers by Dang, Qi, and Ye, have considered
the complexity of computational problems related to Tarski's
theorem \cite{DangQiYe,DangYe-expanding, DangYe18-tech}.  In
particular, in \cite{DangQiYe} the authors provided the
already-mentioned $\log^d N$ algorithm for computing a Tarski fixed
point for a discrete map, $f:[N]^d \rightarrow
[N]^d$, which is monotone under the coordinate-wise order.
In \cite{DangQiYe} they also establish that determining the uniqueness
of the fixed point of a monotone map under coordinate-wise order is
\coNP-hard, and that uniqueness under lexicographic order is also
\coNP-hard (already in one dimension).  
In \cite{DangYe-expanding} the authors studied another variant of the
Tarski fixed point problem, namely computing another fixed point of a monotone
function in an expanded domain where the smallest point is a fixed point;
this variant is \NP-hard (the claim in the paper that this problem is
in $\PPA{}$ has been withdrawn by the authors \cite{DangYe-personal-com19}).
In earlier work, Echenique \cite{Eche07}, studied algorithms
for computing all pure Nash equilibria in supermodular games (and
games with strategic complementaries) whose strategy spaces are 
discrete grids. Of course computing all
pure equilibria is harder than
computing {\em some} pure equilibrium; indeed, we show that computing the
least (or greatest) pure equilibrium of such a supermodular game
is already \NP-hard (Corollary \ref{cor:l-g-equil-np-hard}).
In earlier work Chang, Lyuu, and Ti \cite{CLT08} considered
the complexity of Tarski's fixed point theorem 
over a general finite lattice given via an oracle for its
partial order (not given it explicitly)
and given an oracle for the 
monotone function,
and they observed that the total number of oracle queries required 
to find some fixed point in this model is
linear in the number of elements of the lattice. 
They did not study monotone functions on euclidean grid lattices,
and their results have no bearing on this setting, nor on the black-box setting
for general finite lattices that we consider in section \ref{sec:gen-lattice}.

\noindent {\em Subsequent work:}  Since the publication
of the conference version \cite{EPRY-itcs20}  of this paper in 2020,
a number of subsequent papers 
have established further results, including in particular improved
(but still exponential in the dimension $d$) upper bounds 
on the black-box and white-box complexity of the Tarski fixed point problem.
We provide an overview of 
subsequent work in the Conclusions section.

The rest of this paper is organized as follows:
Section \ref{sec:basics}  provides some basic background and
definitions.   Section \ref{sec:ppad-pls} proves that the $\Tarski$ problem
is contained in both $\PLS$ and $\PPAD$. 
The proof of containment in $\PPAD$ uses a result due to Buss 
and Johnson (\cite{BussJohnson12}) that $\PPAD$ is closed under
polynomial-time Turing reductions; for completeness, we provide
a proof of their result in Appendix \ref{appendix:A}.
Section \ref{sec:lower}
provides a proof that in the black-box oracle model the (expected) number
of queries required by a (randomized) algorithm to find a fixed point 
a 2-dimensional monotone function $f:[N]^2 \rightarrow [N]^2$
is $\Omega(\log^2 N)$.  It  provides two separate proofs, one
for randomized algorithms, and a separate proof for deterministic
algorithms which might be more conducive to generalization.
Section \ref{sec:super-modular} provides brief background
on supermodular games and games with strategic complementarities,
and shows that computing a pure Nash equilibrium for supermodular
games that are succinctly presented (via their (infimum or supremum) best-response 
function)  is inter-reducible to computing a Tarski fixed point
for a monotone function on a grid lattice $[N]^d$.
Section \ref{sec:condon-shapley} shows that computing (approximating) the
value Condon's (Shapley's) stochastic games is P-time reducible to $\Tarski$.
Section \ref{sec:gen-lattice} considers the generalization of
the Tarski fixed point computation problem to the setting of
a general finite lattice of ``dimension'' $d$, and relates query upper bounds for this,
in settings where the lattice is explicitly known, to query upper bounds
for the Tarski fixed point problem on the $d$-dimensional Euclidean grid lattice.   In Section \ref{sec:conclusions}, we conclude
by describing several subsequent results that appeared
after the publication of the conference version of this paper,
which in particular provide better query upper bounds for finding
a Tarski fixed point.  We also list a number of intriguing questions that remain open.  Lastly, in Appendix \ref{appendix:B} we point out that a recent result
(\cite{GHH21}) implies another intriguing problem, the $\Arrival$ problem, is reducible
to $\Tarski{}$.

\section{Basics}

\label{sec:basics}

A partial order $(L, \leq)$ is a {\em complete lattice} if every nonempty subset $S$ of $L$ has a least upper bound (or supremum or join, denoted $\sup S$ or $\lor S$) and a greatest lower bound (or infimum or meet, denoted $\inf S$ or $\land S$) in $L$.
A function $f: L \rightarrow L$ is {\em monotone} if for all pairs of elements $x,y \in L$, $x \leq y$ implies $f(x) \leq f(y)$.
A point $x \in L$ is a {\em fixed point} of $f$ if $f(x)=x$.
Tarski's theorem (\cite{Tarski55}) states that the set $Fix(f)$ of fixed points of $f$ is
a nonempty complete lattice under the same partial order $\leq$;
in particular, $f$ has a greatest fixed point (GFP) and a least fixed point (LFP).

In this paper we will take as our underlying lattice $L$ 
a finite discrete Euclidean grid, which we fix for
simplicity to be the integer grid $[N]^d$, for some positive integers $N, d$, where $[N] =\{ 1,\ldots,N \}$.
Comparison of points is componentwise, i.e. $x \leq y$
iff $x_i \leq y_i$ for all $i=1, \ldots,d$.
We will also consider the corresponding continuous $d$-dimensional box,
$[1,N]^d$, where for $a, b \in \real$ with $a \leq b$, $[a,b] = \{x \in \real \mid a \leq x \leq b\}$.
Both, the discrete and continuous box are clearly complete lattices.

Given a monotone function $f$ on the integer grid $[N]^d$,
the problem is to compute a fixed point of $f$ (any point in $Fix(f)$).
A generally harder problem is to compute specifically the LFP of $f$
or the GFP of $f$. 
We consider the oracle model, in which the function $f$ is
given by a black-box oracle, and the complexity of the algorithm is measured in terms of the number of queries to the oracle.
Alternatively, we also consider an explicit model in which
$f$ is given explicitly by a polynomial-time algorithm
(a polynomial-size Boolean circuit), and then the complexity of the algorithm is measured in the ordinary Turing model.
Note that the number of bits needed to represent a point 
in the domain is $d \log N$,
so polynomial time here means polynomial in $d$ and $n=\log N$.
The number $N^d$ of points in the domain is exponential.

A standard value iteration algorithm provides a simple way to compute
the LFP of $f$: starting from the lowest point of the lattice,
which here is the all-1 vector $1$, apply repeatedly $f$.
This generates a monotonically increasing sequence of points $1 \leq f(1) \leq f^2(1) \leq \ldots$ until a fixed point is reached, which
is the LFP of $f$. In every step of the sequence, at least one coordinate is strictly increased, therefore a fixed point is reached in at most $(N-1)d$ steps.  In the worst case, the process may take that long, which is exponential in the bit size $d \log N$.
Similarly, the GFP can be computed by applying repeatedly $f$
starting from the highest point of the lattice, i.e., from the all-$N$ point, until a fixed point is reached.

Another way to compute some fixed point of a monotone function $f$ (not necessarily the LFP or the GFP) is by a divide-and-conquer algorithm.
In one dimension, we can use binary search: if the
domain is the set $L(l,h)=\{ x \in Z \mid l \leq x \leq h \}$ of integers between the lowest point $l$ and the highest point $h$,
then compute the value of $f$ on the midpoint $m= \lceil (l+h)/2 \rceil$.
If $f(m)=m$ then $m$ is a fixed point; if $f(m) < m$ then recurse on the lower half $L(l,f(m))$, and if
$f(m) > m$ then recurse on the upper half $L(f(m),h)$. The monotonicity of $f$ implies that $f$ maps the respective half interval into itself. 
Hence the algorithm correctly finds a fixed point in at most
$\log N$ iterations, where $N$ is the number of points.

In the general $d$-dimensional case, suppose that the domain is the set of integer points in the box defined by the lowest point $l$ and the highest point $h$, i.e.
$L(l,h)=\{ x \in Z^d \mid l \leq x \leq h \}$. 
Consider the set of points with $d$-th coordinate equal to $m= \lceil (l+h)/2 \rceil$;
their first $d-1$ coordinates induce a $(d-1)$-dimensional lattice
$L'(l,h)=\{ x \in Z^{d-1} \mid l_i \leq x_i \leq h_i, i=1,\ldots d-1 \}$.
Define the function $g$ on $L'(l,h)$ by  
letting $g(x)$ consist of the first $d-1$ components 
of $f(x,m)$. It is easy to see that $g$ is a monotone function
on $L'(l,h)$. Recursively, compute a fixed point $x^*$ of $g$.
If $f_d(x^*,m) = m$, then $(x^*,m)$ is a fixed point of $f$
(this holds in particular if $l=h$).
If $f_d(x^*,m) > m$, then recurse on $L(f(x^*,m),h)$.
If $f_d(x^*,m) < m$, then recurse on $L(l,f(x^*,m))$.
In either case, monotonicity implies that if the algorithm recurses,
then $f$ maps the smaller box into itself and thus has a fixed point in it. 
An easy induction shows that the complexity of this algorithm is $O((\log N)^d)$,  (\cite{DangQiYe}).

Computing the least or the greatest fixed point is in general hard,
even in one dimension, both in the oracle and in the explicit model.

\begin{proposition}\label{lfp:hard}
Computing the LFP or the GFP of an explicitly given polynomial-time monotone function
in one dimension is \NP-hard. In the oracle model, the problem requires
$\Omega(N)$ queries for a domain of size $N$.
\end{proposition}
\begin{proof}
We prove the claim for the LFP; the GFP is similar.
We reduce from the Satisfiability (SAT) problem. Given a Boolean formula $\phi$ in $n$ variables,
let the domain $D=\{ 0, 1, \ldots, 2^n \}$, and define the function $f$ as follows.
For $x \leq 2^n-1$, viewing $x$ as an $n$-bit binary number, it corresponds
to an assignment to the $n$ variables of $\phi$; let $f(x)=x$ if the assignment $x$
satisfies $\phi$, and let $f(x)=x+1$ otherwise. Define $f(2^n)=2^n$.
Clearly $f$ is a monotone function and it can be computed in polynomial time.
If $\phi$ is not satisfiable then the LFP of $f$ is $2^n$, while if $\phi$ is 
satisfiable then the LFP is not $2^n$.

For the oracle model, use the same domain $D$ and let $f$ map every
$x \leq 2^n-1$ to $x$ or $x+1$, and $f(2^n)=2^n$.
The LFP is not $2^n$ iff there exists an $x \leq 2^n-1$ such that $f(x)=x$,
which in the oracle model requires trying all possible $x \leq 2^n-1$.
\end{proof}

In the case of a continuous domain $[1,N]^d$, we may not be able to
compute an exact fixed point, because the exact fixed points may all involve irrational values\footnote{Consider, e.g., the 1-dimensional monotone function $f:[0,1] \rightarrow [0,1]$, given by $f(x) := \frac{1}{2} x^2 + \frac{1}{4}$. 
Clearly, this is a monotone map from the unit interval $[0,1]$ to itself; furthermore, it is easy to check using the quadratic formula that the unique fixed point solution $x \in [0,1]$ satisfying $x = f(x) = (1/2)x^2 + (1/4)$ is $(1-\frac{1}{\sqrt{2}})$, which is irrational.}, and thus we have to be content with
approximation. Given an $\epsilon>0$, an $\epsilon$-approximate fixed point is a point $x$ such that $|f(x)-x| \leq \epsilon$, where we use the $L_{\infty}$ (max) norm, i.e. $|f(x)-x| = \max \{ |f_i(x)-x_i| | i=1,\ldots, d \}$. 
In this context, polynomial time means polynomial in $\log N, d,$ and
$\log (1/ \epsilon)$ (the number of bits of the approximation).
An $\epsilon$-approximate fixed point need not be close to any actual fixed point of $f$. A problem that is generally harder
is to compute a point that approximates some actual fixed point, and an even harder task is to approximate specifically the
LFP or the GFP of $f$.
Value iteration algorithm, starting from the lowest point converges
in the limit to the LFP (and if started from the highest point, 
it converges to the GFP),
but there is no general bound on the number of iterations needed
to get within $\epsilon$ of the LFP (or the GFP).
The algorithm does however reach an $\epsilon$-approximate fixed point within
$N d /\epsilon$ iterations (note, this is exponential in $\log N$ and $\log (1/ \epsilon)$). 

It is easy to see that the approximate fixed point problem for the continuous case reduces to the exact fixed point problem for the discrete case.

\begin{proposition}\label{con-disc}
For any $\epsilon > 0$, the problem of computing an $\epsilon$-approximate fixed point of a given monotone function on the continuous domain $[1,N]^d$
reduces to the problem of computing an exact fixed point for a monotone function on the discrete domain
$[ \: N  \cdot \left\lceil 1/\epsilon \right\rceil \: ]^d$.
\end{proposition}
\begin{proof}
Given the monotone function $f$ on the continuous domain $D_1=[1,N]^d$,
consider the discrete domain 
$D_2= \{ x \in {\mathbb Z}^d | k \leq x_i \leq Nk, i=1, \ldots, d \}$,
where $k= \lceil 1/\epsilon \rceil$, and define
the function $g: D_2 \rightarrow D_2$ as follows.
For every $x \in D_2$, let $g(x)$ be obtained from $kf(x/k)$ by
rounding each coordinate to the nearest integer, with ties broken (arbitrarily) in favor of the ceiling.
Since $f$ is monotone, $g$ is also monotone.
If $x^*$ is a fixed point of $g$, then $k f(x^*/k)$ is within 1/2
of $x^*$ in every coordinate, i.e., $|| k f(x^*/k) - x^* ||_\infty \leq 1/2$, and hence 
$|| f(x^*/k) - x^*/k ||_\infty \leq 1/2k < \epsilon$. Thus  $x^*/k$ is
an $\epsilon$-approximate fixed point of $f$.
\end{proof}

\section{Computing a Tarski fixed point is in \PLS{} $\cap$ 
$\PPAD{}$} 

\label{sec:ppad-pls}

For a monotone function $f: [N]^d
\rightarrow [N]^d$ (with respect to the
coordinate-wise ordering),  we are interested
in computing a fixed point $x^* \in \Fix(f)$,
which we know exists by Tarski's theorem.
We shall formally define this as a discrete total search problem,
using a standard construction to avoid the ``promise'' that $f$ is monotone.

Recall that a general discrete {\em total search problem}
(with polynomially bounded outputs), $\Pi$, 
has a set of {\em valid input instances} $D_\Pi \subseteq \{0,1\}^*$, and
associates with each valid input instance  
$I \in D_\Pi$, a  non-empty set ${\mathcal O}_I
\subseteq \{0,1\}^{p_\Pi(|I|)}$ of {\em acceptable outputs}, 
where $p_\Pi(\cdot)$ is some polynomial.  (So the bit encoding length of
every acceptable output is polynomially bounded in the bit encoding 
length of the input $I$.)

We are interested in the complexity of the following total search 
problem:

\begin{definition}{{\large $\Tarski{}${\rm :}}}

\vspace*{0.1in}

{\bf Input:} A function $f: [N]^d \rightarrow [N]^d$
with $N = 2^n$ for some $n \geq 1$, 
given by a boolean circuit, $C_f$, 
with $(d \cdot n)$ input gates and $(d \cdot n)$ output gates.

{\bf Output:}   Either a (any) fixed point $x^* \in \Fix(f)$,
or else a witness pair of vectors $x,y \in [N]^d$ such that $x \leq y$ and 
$f(x) \not\leq f(y)$.

\end{definition}

Note  $\Tarski$ is a {\em total} search problem:
If $f$ is monotone, it will contain a fixed point in $[N]^d$,
and otherwise it will contain such a witness pair of vectors that
exhibit non-monotonicity. (If it is non-monotone it 
may of course have both witnesses for
non-monotonicity and fixed points; it is important that in this case either kind of output will do.\footnote{Indeed, deciding whether such a function given by a boolean circuit is non-monotone in easily seen to be \NP-complete.})  Note that if we restricted $\Tarski$ to the subproblem $\Tarski_d$, where the dimension $d$ of the input function is assumed to be a fixed constant $d$, not dependent on the input, then for each fixed $d$ the problem $\Tarski_d$ is solvable in polynomial time, using the already mentioned $O(\log^d N)$ time algorithm.

\subsection*{$\Tarski{} \in \PLS{}$}

Recall that a
total search problem, $\Pi$, is in the
complexity class \PLS{} ({\em Polynomial Local Search})
if it satisfies all of the following conditions (see \cite{JPY88,Yan97}):

\begin{enumerate}
\item  For each valid input instance $I \in D_\Pi \subseteq \{0,1\}^*$ 
of $\Pi$, there is an associated non-empty  
set $S_I \subseteq \{0,1\}^{p(|I|)}$ of  {\em solutions},
and an associated {\em payoff function}\footnote{Or, cost function, if 
we were considering local minimization. But here
we focus on local maximization.}, 
$g_I: S_I \rightarrow {\mathbb Z}$.  
For each $s \in S_I$,  
there is an associated set of {\em neighbors}, 
${\mathcal N}_I(s) \subseteq S_I$.

A solution $s \in S_I$ is called 
a {\em local optimum} (local maximum) if  for all $s' \in {\mathcal N}_I(s)$,
$g_I(s) \geq g_I(s')$.   We let ${\mathcal O}_I$ denote the set of all 
local optima for instance $I$.  (Clearly ${\mathcal O}_I$ is 
non-empty, because $S_I$ is non-empty.)

\item  There is a polynomial time algorithm, $A_{\Pi}$,
that given a string $I \in \{0,1\}^*$, decides whether $I$
is a valid input instance $I \in D_{\Pi}$,  and if so outputs 
some solution $s_0 \in S_I$.

\item There is a polynomial time algorithm, $B_{\Pi}$, that
given valid instance $I \in D_{\Pi}$ and a string 
$s \in \{0,1\}^{p(|I|)}$,  
 decides whether  $s \in S_I$, and if so, outputs the
payoff $g_I(s)$. 

\item There is a polynomial time algorithm, $H_\Pi$, that
given valid instance $I \in D_{\Pi}$ and $s \in S_I$,
decides whether $s$ is a local optimum, i.e., whether $s \in {\mathcal O}_{I}$,
and otherwise computes a strictly improving 
neighbor $s' \in {\mathcal N}_I(s)$, such
that $g_I(s') > g_I(s)$.
\end{enumerate}

\begin{theorem} $\Tarski{}  \in  \PLS{}$.
\label{thm:tarski-in-pls}
\end{theorem}

\begin{proof}
Each valid input instance $I_f \in D_{\Tarski} \subseteq
\{0,1\}^*$ of \Tarski{} is an encoding of a 
function $f: [N]^d \rightarrow  [N]^d$
via a boolean circuit $C_f$.
We can view the problem $\Tarski{}$
as a polynomial local search problem, as follows:

\begin{enumerate}
\item Define the set of ``solutions'' associated with valid input $I_f$ to be 
the disjoint union $S_{I_f}= S'_{I_f} \cup S''_{I_f}$,
where  $S'_{I_f} = 
\{ x \in [N]^d \mid  x \leq f(x) \}$
and $S''_{I_f} = \{ (x,y) \in [N]^d \times [N]^d \mid  x \leq y  \wedge
f(x) \not\leq f(y) \} $.
Clearly,  $\Fix(f) \subseteq S'_{I_f} \subseteq S_{I_f}$.
Let the payoff function $g_{I_f}: S_{I_f} \rightarrow {\mathbb Z}$,
be defined as follows.  For $x \in S'_{I_f}$,
$g_{I_f}(x) := \sum_{i=1}^d x_i$;  for $(x,y) \in S''_{I_f}$,
$g_{I_f}(x,y) := dN+1$.
We define the neighbors of solutions as follows. 
For any $x \in S'_{I_f}$,  if $f(x) \leq f(f(x))$ then
let the neighbors of $x$ be 
the singleton-set ${\mathcal N}_{I_f}(x)
:= \{ f(x) \}$.   Note that in this case again $f(x) \in S'_{I_f}$.
Otherwise, if $f(x) \not\leq f(f(x))$, then
let ${\mathcal N}_{I_f}(x)
:= \{ (x, f(x)) \}$. Note that in this case $(x,f(x)) \in S''_{I_f}$,
since $f(x) \not\leq f(f(x))$. 
  For $(x,y) \in S''_{I_f}$, let  ${\mathcal N}_{I_f}(x,y) 
:= \emptyset$ be the empty set.
Thus, the set of local optima is by definition 
${\mathcal O}_{I_f} = \{ x \in S'_{I_f} \mid   \sum^d_{i=1} x_i \geq 
\sum_{i=1}^d f_i(x) \}  \cup S''_{I_f} $.

Observe that in fact ${\mathcal O}_{I_f} = \Fix(f) \cup S''_{I_f}$.
Indeed,  if $x \in {\mathcal O}_{I_f}$ then $x \in S'_{I_f}$ 
meaning $x \leq f(x)$, and also
$\sum^d_{i=1} x_i \geq \sum_{i=1}^d f_i(x)$.  But this is
only possible if $f(x) = x$, i.e., $x \in \Fix(f)$.
Likewise, if $(x,y) \in {\mathcal O}_{I_f}$ then $(x,y) \in S''_{I_f}$.
On the other hand, if $x \in \Fix(f)$, then clearly $x \in S'_{I_f}$
and $\sum^d_{i=1} x_i = \sum_{i=1}^d f_i(x)$, hence $x \in
{\mathcal O}_{I_f}$.

\item There is a polynomial time algorithm $A_{\Tarski}$ that,
given a string $I_f \in \{0,1\}^*$ first determines whether
this is a valid input instance, by checking that it suitably encodes
a boolean circuit (straight-line program) 
$C_f$ with $(d \log N )$ input gates and the same number of
output gates, 
and thereby defines a function $f:[N]^d \rightarrow [N]^d$.
If the input is a valid instance, then $A_{\Tarski}$
outputs a solution $s_0 \in S_{I_f}$,
by just letting $s_0 := {\mathbf 1} \in [N]^d$
be the all 1 vector.
Clearly,  ${\mathbf 1} \leq f({\mathbf 1})$, so indeed 
${\mathbf 1} \in S'_{I_f} \subseteq S_{I_f}$.

\item There is a polynomial time algorithm $B_{\Tarski}$ that,
given a valid instance $I_f$ and given
a string $s \in \{0,1\}^*$,
first decides
whether $s \in S_{I_f}$.
It does so as follows:  if $s$
is (a binary encoding of) $x \in [N]^d$, 
then $B_{\Tarski}$  computes
$f(x)$ using the given boolean
circuit $C_f$ (encoded in instance $I_f$), and checking whether $x \leq f(x)$.
If instead $s = (x,y) \in [N]^d \times [N]^d$,  then it checks whether
$(x,y)$ is a witness of non-monotonicity, by computing $f(x)$ and
$f(y)$ using $C_f$, and checking that both $x \leq y$ and $f(x) \not\leq f(y)$
hold.

If $s \in S_{I_f}$, 
the algorithm can also easily output
the value of the objective $g_{I_f}(s)$.
Namely, if $s = x \in S'_{I_f}$, then
$g_{I_f}(s) := \sum_{i=1}^d x_i$, and if $s = (x,y) \in S''_{I_f}$
then $g_{I_f}(s) := dN + 1$.

\item Finally, there is a polynomial time algorithm $H_{\Tarski}$
that, given an instance $I_f \in D_{\Tarski}$ and a solution $s \in S_{I_f}$,
decides whether $s \in {\mathcal O}_{I_f} = \Fix(f) \cup S''_{I_f}$,
and otherwise outputs $s' \in {\mathcal N}_{I_f}(s)$,
such that $g_{I_f}(s') > g_{I_f}(s)$.
Firstly, if $s = (x,y) \in S''_{I_f}$ (which we can check as in the prior
item), then clearly $s \in {\mathcal O}_{I_f}$
and there is nothing more to do.
If on the other hand $s = x \in S'_{I_f}$,
the algorithm uses the given circuit $C_f$ to compute $f(x)$,
checks first whether $f(x)=x$. If so, we are done.
If not, it checks whether $f(x) \leq f(f(x))$  and if so 
it outputs $x' := f(x)$. 
In this case, since $x \leq f(x)$ and $x \neq f(x)$, we indeed have
strictly improved the objective:
$g_{I_f}(x') = \sum_{i=1}^d f_i(x) >  \sum_{i=1}^d x_i = g_{I_f}(x)$.
Finally, if $f(x) \not\leq f(f(x))$ it outputs the pair 
$(x,f(x))$.   Note that in this case $(x,f(x)) \in {\mathcal N}_{I_f}(x)$,
and that we do strictly improve the objective value, since
$g_{I_f}((x,f(x)) = dN+1 > \sum_{i=1}^d N \geq  \sum_{i=1}^d x_i
= g_{I_f}(x)$.
\end{enumerate}

We have thus shown that $\Tarski{}$ satisfies
all the conditions of being in $\PLS{}$.
\end{proof}

\subsection*{$\Tarski{} \in \PPAD$}

To show that $\Tarski \in \PPAD$,  we
first show that $\Tarski \in \FP^{\PPAD}$ meaning that the total search
problem $\Tarski{}$ can be solved by a polynomial time algorithm,
${\mathcal M}$, with oracle access to $\PPAD$.  The algorithm
${\mathcal M}$ should take an input $I_f \in \{0,1\}^*$, and firstly
decide whether it is a valid instance $I_f \in D_\Tarski$, and if so
it can make repeated, adaptive, calls to an oracle for solving a
$\PPAD$ total search problem.  After at most polynomial time (and
hence polynomially many such oracle calls) as a function of the input
size $|I_f|$, ${\mathcal M}$ should output either an integer vector $x
\in \Fix(f)$, or else output a pair of vectors $x,y \in [N]^d$ with $x
\leq y$ and $f(x) \not\leq f(y)$, which witness non-monotonicity of
the function $f: [N]^d \rightarrow [N]^d$ defined by the input
instance $I_f$.

Once we have established that $\Tarski \in \FP^{\PPAD}$, the
fact that $\Tarski \in \PPAD$ follows as a simple corollary,
using a result due to Buss and Johnson \cite{BussJohnson12}
that $\PPAD$ is closed under polynomial-time Turing
reductions.

There are a number of equivalent ways to define the total search 
complexity class $\PPAD$.   
Rather than give the original definition (\cite{Pap94}),
we will use an equivalent characterization of $\PPAD$
(a.k.a., $\mathsf{linear}$-$\mathsf{FIXP}$) from 
\cite{EY2010} (see  section 5 of \cite{EY2010}).
Informally, according to this characterization, 
a discrete total search problem, $\Pi$, is in $\PPAD$ if and only if
it can be reduced in P-time to computing a Brouwer fixed point
of an associated ``polynomial piecewise-linear'' continuous function that maps a non-empty
convex polytope to itself.
More formally, $\Pi$ is in $\PPAD$ if it satisfies
all of the following conditions:

\begin{enumerate}
\item Each valid instance $I \in D_{\Pi} \subseteq \{0,1\}^*$ can be associated
with a ``polynomial-time definable'' (see below)
piecewise-linear continuous function $F_I: W(I) \rightarrow W(I)$. 
Here $W(I) \subseteq {\mathbb R}^{d_I}$ is a non-empty (rational) convex polytope.

\item There is a polynomial time algorithm, $R^1_\Pi$, that, given
a string $I \in \{0,1\}^{d_I}$, first decides whether $I$ is a valid instance in $D_{\Pi}$
of $\Pi$,  and if so, outputs a rational matrix $A^I \in \rat^{m_I \times d_I}$
and a rational vector $b^{I} \in \rat^{d_I}$, 
such that $W(I) = \{ x \in \real^{d_I} \mid  A^I x \leq b^I \}$ is a non-empty convex polytope.

\item There is a polynomial time oracle algorithm, $R^2_\Pi$,
that ``computes'' the piecewise-linear function $F_I$ in the following sense.

For any real vector $x \in \real^{d_I}$, consider an oracle
$O_x$ with the following property:
when  $O_x$ is called with a  
rational vector $a \in \rat^{d_I}$ and a rational value $c \in \rat$,
then 
the oracle outputs $O_x(a,c) = {\mathrm{TRUE}}$ if $a^T x \leq c$,
and otherwise it outputs $O_x(a,c) = {\mathrm{FALSE}}$.

$R^2_\Pi(I,O_x)$, runs in time polynomial in $|I|$,
and hence makes $poly(|I|)$ many calls to the oracle $O_x$,
for any $x \in \real^{d_I}$.
When given as input a valid instance $I \in D_\Pi$ and oracle
access to  $O_x$ for some $x \in \real^{d_I}$,  
$R^2_\Pi(I,O_x)$
outputs  ``No'' if
$x \not\in W(I)$, and otherwise, if $x \in W(I)$, then it outputs a rational 
matrix $C \in \rat^{d_I \times d_I}$, and a rational vector $C' \in \rat^{d_I}$,
such that $F_I(x) = C x + C' \in W(I)$.   (Note that since
$R^2_\Pi$ runs in polynomial time in the input size $|I|$, 
the bit encoding
sizes of the coefficients in $C$ and $C'$ are polynomial in $|I|$.)

Note that in this sense $R^2_{\Pi}$ does indeed define the piecewise-linear function
$F_I : W(I) \rightarrow W(I)$.   Specifically,  for $x \in W(I)$, 
the sequence of (polynomially many) 
oracle queries  made by $R^2_{\Pi}(I,O_x)$  defines
a system of linear inequalities (with rational, polynomially bounded, coefficients) 
satisfied by $x$ which define a ``piece'' or ``cell''
such that $x \in {\mathcal C}_x \subseteq W(I)$, and
such that $F_I$ is linear on ${\mathcal C}_x$; specifically 
such that for any $y \in {\mathcal C}_x$,
$F_I(y) =  C y + C'$.

\item  There is a polynomial time algorithm $R^3_\Pi$ that, given an instance $I \in D_\Pi$,
and given any rational fixed point $x^* \in \Fix(F_I) \cap \rat^{d_I}$,  
outputs an acceptable output in
${\mathcal O}_I$ for the instance $I$ of the total search problem $\Pi$.
\end{enumerate}

By Brouwer's theorem, the set
 $\Fix(F_I) = \{ x \in W(I) \mid  F_I(x) = x \}$ of fixed points of $F_I$ is non-empty.    
Moreover,
because of the ``polynomial piecewise-linear'' nature of $F_I$, 
$\Fix(F_I)$ must also contain a 
{\em rational} fixed point
$x^* \in {\mathbb Q}^{d_I}$, with polynomial bit complexity as a function of $|I|$ 
(see \cite{EY2010}, Theorem 5.2).   See \cite{EY2010}, section 5, for more details on
this characterization of \PPAD.

Given two vectors $l \leq h \in {\mathbb Z}^d$,  let 
$L(l,h) = \{x \in {\mathbb Z}^d \mid  l \leq x \leq h \}$, and
let $B(l,h) =\{ x \in {\mathbb R}^d \mid l \leq x \leq h \}$.

\medskip

\begin{theorem} 
$\Tarski \in \FP^{\PPAD}$.
\label{thm:tarski-in-P-to-PPAD}
\end{theorem}
\begin{proof}
Suppose we are given an instance $I_f \in D_\Tarski$ of $\Tarski$,  corresponding
to a function $f: [N]^d \rightarrow [N]^d$  (given by a boolean circuit $C_f$).

Let $a = {\mathbf 1} \in {\mathbb Z}^d$, and $b = {\mathbf N} \in {\mathbb Z}^d$, denote
the all $1$, and all $N$, vectors respectively.
Thus $L(a,b) = [N]^d$.
We first extend the discrete function $f: L(a,b) \rightarrow L(a,b)$ to a (polynomial piecewise-linear) continuous function 
$f': B(a,b) \rightarrow B(a,b)$, by a suitable linear interpolation.
By Brouwer's theorem,  $f'$ has a fixed point in $B(a,b)$.
However,  $f'$ may have non-integer fixed points that do not
correspond to (and are not close to)  any fixed point of $f$  (indeed, since we do
not apriori know that $f$ is monotone, there may not be any integer fixed points). 

Nevertheless, we will show that finding any such fixed point of $f'$ allows us
to make progress (via a divide and conquer binary search), towards either finding 
a discrete fixed point of $f$ (if it is monotone), or finding witnesses for a violation
of monotonicity of $f$.

We now define $f'$ in detail.  
Consider the following simplicial decomposition\footnote{Known as
{\em Freudenthal's simplicial division} \cite{Freud42}.}
 of $B(a,b)$. 
For each $i \in [d]$, let $e^i \in \{0,1\}^d$ denote
the standard unit vector with $0$'s in every coordinate except a $1$ in the $i$'th coordinate.
For each integer vector $y \in L(a,b-{\mathbf 1})$,  
and for every permutation $\pi = (\pi_1,\ldots,\pi_d)$ of $[d]$,
define the subsimplex $S^{y,\pi}$ as the convex hull of the following $d+1$
(affinely independent) vertices  $y^0, \ldots,y^{d+1} \in L(a,b)$, given by
$y^0 = y$,  and for $i \in \{1,\ldots, d\}$, $y^i = y^{i-1} + e^{\pi_i}$.

The union of all $d!$ simplices $\{S^{y,\pi} \mid \pi \ \mbox{is a permutation of $[d]$}\}$ 
constitutes a simplicial subdivision of the d-cube $B(y, y+{\mathbf 1})$,
and the union of all such simplices, for all $y \in L(a,b-{\mathbf 1})$
constitutes a simplicial subdivision of $B(a,b)$.
Note the following important property of this simplicial subdivision, which we 
exploit:  the vertices $y^{0}, y^{1}, \ldots, y^{d}$ of each subsimplex $S^{y,\pi}$ are totally
ordered with respect to coordinate-wise order: $y^0 \leq y^{1} \leq y^{2} \leq \ldots \leq y^{d}$.

Given this simplicial subdivision of $B(a,b)$,  we define $f': B(a,b) \rightarrow B(a,b)$
so that it linearly interpolates $f$ inside each subsimplex $S^{y,\pi}$.
Specifically,  
for any point $x \in S^{y,\pi}$,  there is a unique vector 
$\lambda = (\lambda_0, \lambda_1, \ldots, \lambda_{d}) \in [0,1]^{d+1}$, such that $\sum_{j=0}^d 
\lambda_j = 1$,  and  such that $x = \sum_{j=0}^d \lambda_j y^j$.
We define $f'(x) := \sum_{j=0}^d \lambda_j f(y^j)$.
Note that $f'$ agrees with $f$ on integer points in $L(a,b)$.
Also if $x$ belongs to several $x \in S^{y,\pi}$ (i.e. lies on some common
faces of the subsimplices), then only the common vertices
will have nonzero coefficients in any subsimplex, thus they all yield the same value
for $f'(x)$, and hence $f'(x)$ is consistently defined. 

Our next task is to show that computing a rational fixed point of $f'(x)$ is in $\PPAD$, which
will allow us to use the $\PPAD$ oracle to find such a rational fixed point.
Applying the definition of $\PPAD$ we have given above, all we need to do is to 
specify a polynomial time oracle algorithm that, given oracle access $O_x$ to
some $x \in \real^d$, can first locate the subsimplex $S^{y,\pi}$ such that
$x \in S^{y,\pi}$ (or report that $x$ is not in the domain $B(a,b)$), and 
then compute the matrix $C$ and vector $C'$ that specify the affine transformation
such that $f'(x) = Cx + C'$.    It was explained in \cite{EY2010}
(see page 2583, second paragraph) how to do this
for a standard simplicial decomposition, and essentially the same
approach works for the simplicial decomposition we are using here.

Thus, $f'$ is a polynomial piecewise-linear Brouwer function, 
and we can compute a 
rational fixed point  $x^* \in \Fix(f')$ for it in $\PPAD$.
If $x^*$ is an integer vector, we are done: we have found a fixed point of $f$.

Suppose, on the other hand, that the computed fixed point $x^*$ of $f'$ is non-integer
in some coordinate.  It is 
still useful.  
Consider the cell $C \subseteq S^{y,\pi}$, 
defined as the convex hull of the unique subset $Y' = \{y^{j_1}, y^{j_2}, \ldots, y^{j_k}\}$ 
of the vertices $Y= \{y^0, \ldots,y^{d}\}$
of $S^{y,\pi}$,
such that $C$  contains $x^*$ in its {\em strict} interior.
In other words, $x^* = \sum_{r=1}^k \lambda_{j_r} y^{j_r}$, such that $0 < \lambda_{j_r} < 1$
for all $r \in \{1,\ldots,k\}$.
Let $u= y^{j_t}$ be the 
maximum vertex of $C$, and let
$v= y^{j_q}$ be the minimum vertex of $C$  (the vertices of $C$ are ordered since
they are a subset of the vertices of
$S^{y,\pi}$). 

Suppose that $f$ is monotone and $f(u)_i < u_i$ for some coordinate $i$.
Then $f(u)_i \leq u_i-1$ because $f(u)_i$ is an integer.
Furthermore, for all vertices $y^{j_r}$ of $C$, since $y^{j_r} \leq u = y^{j_t}$, 
we must also have $f_i(y^{j_r}) \leq f_i(u) \leq u_i-1 \leq y^{j_r}_i$
(where the last inequality holds because two vertices of $C$  differ
 in any given coordinate by at most 1).
But we have $\sum_{r=1}^k \lambda_{j_r} y^{j_r}_i = x^*_i
= f'(x^*)_i = \sum_{r=1}^k \lambda_{j_r} f(y^{j_r})_i$,
which is impossible, since $y^{j_r}_i \geq f(y^{j_r})_i$
for every $r$, and $y^{j_t}_i > f(y^{j_t})_i$,
and $\lambda_{j_t} > 0$.
Thus, since $x^*$ is a fixed point of $f'$,
it cannot be the case that $f$ is monotone and $f(u)_i < u_i$ for
some coordinate $i$.
Therefore, if $f$ is monotone, then
$f(u) \geq u$ (in all coordinates).
For a completely analogous reason, if $f$ is monotone, we also have $f(v) \leq v$.

Suppose, on the other hand we either find that $f(u) \not\geq u$,  or
that $f(v) \not\leq v$.   Then necessarily, it must be the case that there 
are a pair of vertices $y^{j_b}, y^{j_e}$ of the cell $C$ containing $x^*$ in its interior,
such that $y^{j_b} \leq y^{j_e}$ but $f(y^{j_b}) \not\leq f(y^{j_e})$.
So, in this case, we examine all such pairs to find such a pair, we 
halt and output $(y^{j_b}, y^{j_e})$ as a witness pair for
the non-monotonicity of $f$.

Assume on the other hand that $f(u) \geq u$ and $f(v) \leq v$.
Note that in that case, if $f$ is monotone, 
then it maps the sublattice $L(u,b)$ to itself,
and it also maps the disjoint sublattice $L(a,v)$ to itself.
Thus, if  $f$ is monotone, $f$ must have an integer fixed point in both 
$L(a,v)$ and $L(u,b)$.

So, we can choose the smaller of these two sublattices, 
consider the function $f$ restricted to that sublattice,
and continue recursively to find a fixed point in that sublattice (if $f$ is monotone)
or a violation of monotonicity.
If $f$ is not monotone, it is possible that it maps some points in the sublattice
$L(a,v)$ (or $L(u,b)$) to points outside. 
Therefore, in the recursive call for the sublattice,
when we define the piecewise-linear function $f'$ on the corresponding
box $B(a,v)$ (or $B(u,b)$) we take the maximum with $a$ and minimum with $v$
(or $u$ and $b$ respectively),
i.e., threshold it, so that it maps the box to itself, and hence it
is a Brouwer function.
When the $\PPAD{}$ oracle gives us back a fixed point $x^*$ for this
(possibly thresholded) function $f'$, we find the vertices $y^{j_r}$ of 
the cell $C$ that contains $x^*$ in its strict interior
(i.e. the ones that have nonzero coefficients in the convex combination)
and test if $f$ maps all of them within the current box.
If this is not the case then we get a violation of monotonicity:
Suppose wlog that the current box is $B(a,v)$ (similarly if it is $B(u,b)$).
If $f(y^{j_r}) \not\geq a$ then $(a,y^{j_r})$ is a violating pair 
because $f(a) \geq a$; if $f(y^{j_r}) \not\leq v$ then
$(y^{j_r},v)$ is a violating pair because $f(v) \leq v$.
Thus, if $f(y^{j_r})$ lies outside the current box, then
we return the discovered violating pair and terminate.
Otherwise, the thresholding did not affect the $f(y^{j_r})$ and $f'(x^*)$
and we proceed as explained above.
  
Every iteration decreases the
total number of points in our current lattice by a factor of $2$, from the number of points $N^d$
in the original lattice $L(a,b)$.
So after at most $\log(N^d) = d \log N$ iterations, we either
find a fixed point of $f$, or we find a witness pair of integer vectors that witness the non-monotonicity
of $f$, or else in the trivial base case where we have 
reduced the domain under consideration to a singleton set $L(a',b')$ with $a'=b'$ (with the inductively maintained property that $a' \leq f(a')$ and $f(b') \leq b'$), we necessarily have $a' = f(a')$, i.e., we have found a fixed point.
\end{proof}

\begin{corollary}
$\Tarski \in \PPAD$.
\label{cor:tarski-in-PPAD}
\end{corollary}
\begin{proof}
This follows immediately from Theorem \ref{thm:tarski-in-P-to-PPAD},
combined with a result due to Buss and Johnson 
(\cite{BussJohnson12}, Theorem 6.1),  
that $\PPAD$ is closed under polynomial-time Turing 
reductions.
Indeed, Theorem 6.1 in \cite{BussJohnson12}
asserts that several different total search complexity classes are 
closed under P-time Turing reductions, including
$\PPAD$, $\PLS$, and $\PPA$.\footnote{Whereas, in \cite{BussJohnson12} they left open whether $\PPP$
is closed under P-time Turing reductions, and conjectured that it is not. 
More recently, \cite{FGPR24} provide 
evidence in the black-box setting, strongly suggesting that $\PPP$ may not be closed under P-time Turing reductions.}
However, their proof of Theorem 6.1 only shows this for the case of $\PLS$, 
and they say ``the others are similar and are left to the reader''.
Likewise,  Johnson's PhD thesis \cite{Johnson2011}, which contains and expands on the contents of \cite{BussJohnson12},  asserts the same theorem (\cite{Johnson2011}, Theorem 6.1.1) but again
only provides a proof for the case of $\PLS$ and says ``the others are similar and are left
to the reader''.
For the sake of completeness, in Appendix \ref{appendix:PPAD-closed-under-Turing-red} 
we provide a proof that $\PPAD$ is closed under polynomial time Turing reductions.
\end{proof}

\section{The 2-dimensional lower bound}

\label{sec:lower}

Consider a monotone function defined on the $ N \times N$ grid
$f:[N]^2\mapsto [N]^2$.  Let $A$ be any (randomized) 
{\em black-box  algorithm} for finding a fixed point of the function by computing a
sequence of queries of the form $f(x,y) = ?$; $A$ can of course be
{\em adaptive} in that any query can depend in arbitrarily complex
ways on the answers to the previous queries.  For example, the
divide-and-conquer algorithm described in the introduction is a black-box algorithm.  
The following result tells us this algorithm is
essentially optimal for two dimensions.

\begin{theorem}
\label{thm:main-2d}Given black-box access to a monotone function $f:\left[N\right]^{2}\rightarrow\left[N\right]^{2}$, any (randomized) algorithm for
finding a fixed point of $f$ requires $\Omega(\log^{2} N)$
queries  (in expectation).
\end{theorem}

To prove this theorem we construct 
a ``hard distribution'',  $\mathcal{D}_{N,2}$,
on monotone functions, $f:[N]^2 \rightarrow [N]^2$, 
and show that the expected
number of queries required by any (randomized) algorithm
on a function sampled from the distribution $\mathcal{D}_{N,2}$
is at least $\Omega(\log^2 N)$.

The hard distributions $\mathcal{D}_{N,2}$ will be
over a class of monotone functions which 
we call ``herringbones", an example of 
which is depicted in Figure \ref{fig:2d-f}.
For notational convenience only, throughout our proof below we assume that $N^{1/4}$
(and hence also $\sqrt{N}$) is an integer.
This is without loss of generality:  for any positive integer
$N$,   let $N' \leq N$ be the largest integer no greater than $N$
such that $(N')^{1/4}$ is an integer.   Note that when $N$ is  large enough
($N \geq 1000$ suffices),  $N' \geq N/2$.
Given a herringbone function $f':[N']^d \rightarrow [N']^d$
(or a distribution $\mathcal{D}_{N',2}$ over such herringbone functions), 
we can extend $f'$ to the monotone function $f:[N]^2 \rightarrow [N]^2$
(or the corresponding distribution $\mathcal{D}_{N,2}$ over such monotone functions $f$,  respectively) where $f$ is
defined by letting $f(x,y) :=  f'(\min(x,N'), \min(y,N'))$, 
for all $(x,y) \in [N]^2$.  
Note that 
$f$ has 
exactly the same fixed points as $f'$, and any lower
bound for finding a fixed point for a function from the distribution
$\mathcal{D}_{N',2}$ also holds for the distribution $\mathcal{D}_{N,2}$.  We will show a lower bound of $\Omega(\log^2 N')$ 
over $\mathcal{D}_{N',2}$ when
$(N')^{1/4}$ is an integer.  Since $N' \geq N/2$ for large enough $N$, we thus also obtain a lower bound of $\Omega(\log^2 N)$ over 
$\mathcal{D}_{N,2}$.

\subsubsection*{Preliminaries: lower bound for binary search}

Before proceeding, we establish a simple lemma on the expected
number of queries required to find a random unknown 
element $i^* \in [n]$,  when the only thing
available to the algorithm is a comparison oracle that
compares the unknown
element $i^*$ with a chosen element $j \in [n]$, chosen by the algorithm,
and returns which of the three cases (A) $i^* > j$, (B) $i^* = j$, or
(C) $i^* < j$ holds.

\begin{claim}[Cost of binary search with partial probability of success]\label{claim:binary}
  Let $i^*$ be chosen uniformly at random from $[n] = \{1,\ldots, n\}$. 
  For any randomized
  algorithm with access only to the above oracle for comparison with $i^*$,
  if 
   the algorithm succeeds with probability $p > 0$ to
  find $i^*$,
assuming the algorithm always makes at least one query, 
  it must make $\Omega(p \log(n))$ queries in expectation.
\end{claim}

Note that the claim extends to the case where we condition on the
algorithm making at least one query. (The decision of whether or not
to make the first query is independent of the input.)

\begin{proof}
Fix a random seed $r$ for the internal coin flips of the randomized 
algorithm, and let $1 \leq t_r  \leq n$ be
the number of elements $j \in [n]$ such that if $i^*=j$ then
the algorithm succeeds with
random seed $r$.\footnote{Note that $t_r \geq 1$, because we assume
the algorithm always makes at least one query.}  Consider the ternary decision tree, $T_r$, corresponding to
the algorithm's queries, given random seed $r$. Then with probability
$p_r := \frac{t_r}{n}$, the uniformly random unknown input 
$i^* \in [n]$  is such that the algorithm reaches one of the
at least $t_r$ success leaves. 
 
It follows that at least $(\frac{2}{3} t_r)$ success
leaves in $T_r$ must have depth at least $\log_3(t_r)$.
Hence the 
average depth of success leaves in $T_r$ is at least $\Omega(\log(t_r)) = \Omega(\log(p_r n))$.   

Hence, the expected number of comparisons
conditioning on random seed $r$ is at least $\Omega(p_r \log(p_r n))$.
Taking expectation over all $r$, we have that
\begin{eqnarray} 
\E_r[p_r\log(p_r n)]  &  = &    \E_r [p_r (\log(p_r) + \log (n)) ]  \nonumber \\
                      &  = &   \E_r [p_r \log(p_r) ] + \E_r [p_r] \log(n)  \quad \quad \quad \quad \quad \quad \quad  \quad \quad \mbox{(by linearity of expectation)} \nonumber \\
                      & \geq & \E_r[p_r] \log(\E_r[p_r]) 
                                + \E_r[p_r] \log(n) \label{eqn:inequality-expect-bin-search} \\
                      & = &  \E_r[p_r] \log (\E_r[p_r] \log(n)) 
                      \nonumber \\
                      & = &   p \log(p n)  \nonumber
\end{eqnarray}
where inequality (\ref{eqn:inequality-expect-bin-search}) follows from Jensen's inequality
and the convexity of the function
$x \log x$ over the domain $x > 0$.

Since we assume that the algorithm makes at least one query, we can strengthen the lower bound to $\Omega(\max\{1,p\log(p n)\})$. Finally, 
since for all $p \in [0,1]$ we have\footnote{
For concreteness, here we assume $\log(x)$ denotes log base 2.} $p \log (p) \in [-0.6, 0]$,
notice that whenever $p\log(p n) \geq 1$ 
we have $p \log(p n) = 
p \log(p) + p \log(n) = \Theta(p\log(n))$.
\end{proof}

\subsubsection*{The basic construction}

Given a monotone path from $\left(1,1\right)$ to $\left(N,N\right)$
on the $N \times N$ grid graph,
where each two consecutive nodes along the path must 
differ by exactly 1 in one coordinate and be equal in the other coordinate
(i.e., have Manhattan distance 1),
and given a point $\left(i^{*},j^{*}\right)$
on the path, we construct $f$ as follows: 
\begin{itemize}
\item We let $\left(i^*,j^*\right)$ be the unique fixed point of $f$, i.e.
$f\left(i^{*},j^{*}\right)\triangleq\left(i^{*},j^{*}\right)$. 
\item At all other points on the path, $f$ is directed towards the fixed
point: for a point $\left(x,y\right)$ on the path that is dominated by
$\left(i^{*},j^{*}\right)$, meaning $\left(x,y\right) \leq \left(i^{*},j^{*}\right)$, we let $f(x,y)$ be the
next point on the path, i.e. $f(x,y)=\left(x+1,y\right)$
or $f(x,y)=\left(x,y+1\right)$. Similarly, for a point
$(x,y)$ that is on the path and dominates $\left(i^{*},j^{*}\right)$,
we let $f(x,y)$ be the previous point on the path.
\item For all points outside the path, $f$ is directed towards the path as follows.
Observe that the path partitions $\left[N\right]^{2}$ into three
(possibly empty) subsets: below the path, the path, and above the
path. For a point $\left(x,y\right)$ below the path, we set $f\left(x,y\right)\triangleq\left(x-1,y+1\right)$.
Similarly, for a point $\left(x,y\right)$ above the path, $f\left(x,y\right)\triangleq\left(x+1,y-1\right)$.
\end{itemize}

An example of such a function $f:[5]^2 \rightarrow [5]^2$ is given in Figure \ref{fig:2d-f}.  The following is easily verified:

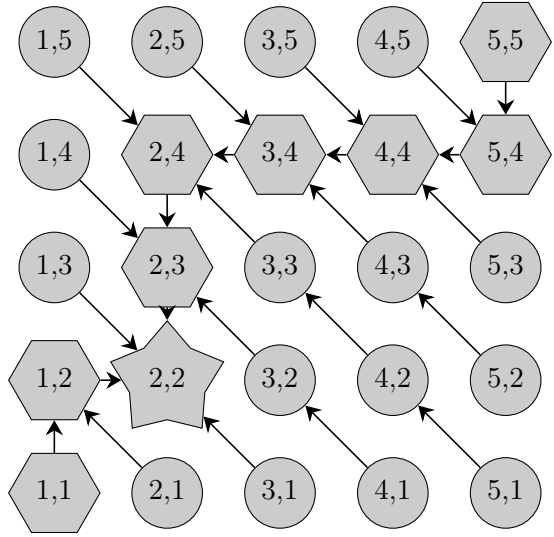
\begin{figure}
\begin{center}
\begin{tikzpicture}[darkstyle/.style={draw,fill=gray!40,minimum size=17}]
\tikzstyle{offpath} = [darkstyle, circle]
\tikzstyle{pathnode} = [darkstyle, regular polygon,regular polygon sides=6]
\tikzstyle{fixedpoint} = [darkstyle, star]

  \foreach \x in {1,...,5}
    \foreach \y in {1,...,5} 
       \node [offpath]  (\x\y) at (1.5*\x,1.5*\y) {\x,\y};

\def\x{1}
\def\y{1}
       \node [pathnode]  (\x\y) at (1.5*\x,1.5*\y) {\x,\y};

\def\x{1}
\def\y{2}
       \node [pathnode]  (\x\y) at (1.5*\x,1.5*\y) {\x,\y};
        
\def\x{2}
\def\y{2}
       \node [fixedpoint]  (\x\y) at (1.5*\x,1.5*\y) {\x,\y};

\def\x{2}
\def\y{3}
       \node [pathnode]  (\x\y) at (1.5*\x,1.5*\y) {\x,\y};

\def\x{2}
\def\y{4}
       \node [pathnode]  (\x\y) at (1.5*\x,1.5*\y) {\x,\y};

\def\x{3}
\def\y{4}
       \node [pathnode]  (\x\y) at (1.5*\x,1.5*\y) {\x,\y};

\def\x{4}
\def\y{4}
       \node [pathnode]  (\x\y) at (1.5*\x,1.5*\y) {\x,\y};

\def\x{5}
\def\y{4}
       \node [pathnode]  (\x\y) at (1.5*\x,1.5*\y) {\x,\y};

\def\x{5}
\def\y{5}
       \node [pathnode]  (\x\y) at (1.5*\x,1.5*\y) {\x,\y};

\tikzstyle{arrowstyle} = [-{Stealth[angle'=60]},semithick]

\draw [arrowstyle] (11)--(12);
\draw [arrowstyle] (12)--(22);
\draw [arrowstyle] (55)--(54);
\draw [arrowstyle] (54)--(44);
\draw [arrowstyle] (44)--(34);
\draw [arrowstyle] (34)--(24);
\draw [arrowstyle] (24)--(23);
\draw [arrowstyle] (23)--(22);

\def\y{5}
 \foreach \x in {1,...,4}
{\pgfmathtruncatemacro{\xx}{\x + 1}
\pgfmathtruncatemacro{\yy}{\y -1 }
\draw [arrowstyle] (\x\y)--(\xx\yy);}

\def\x{1}
 \foreach \y in {3,4}
{\pgfmathtruncatemacro{\xx}{\x + 1}
\pgfmathtruncatemacro{\yy}{\y -1 }
\draw [arrowstyle] (\x\y)--(\xx\yy);}

 \foreach \x in {3,...,5}
  \foreach \y in {1,...,3}
{\pgfmathtruncatemacro{\xx}{\x - 1}
\pgfmathtruncatemacro{\yy}{\y +1 }
\draw [arrowstyle] (\x\y)--(\xx\yy);}

 \def \x{2}
  \def\y{1}
{\pgfmathtruncatemacro{\xx}{\x - 1}
\pgfmathtruncatemacro{\yy}{\y +1 }
\draw [arrowstyle] (\x\y)--(\xx\yy);}

\end{tikzpicture}

\end{center}
\caption{A 2-dimensional ``{\em herringbone}'' monotone function.
\label{fig:2d-f}}
\end{figure}

\begin{claim}
For any choice of path and point $\left(i^{*},j^{*}\right)$ on the
path, $f$ constructed as above is monotone.
\end{claim}

\paragraph*{Choosing the fixed point}

In our hard distribution, once we fix a path, we choose $\left(i^{*},j^{*}\right)$
uniformly at random among all points on the path. 
\begin{claim}[Finding the fixed point requires $\log(N)$ queries to the path]
\label{claim:points-on-path}Given the path and given 
oracle access to $f$,
any (randomized) algorithm that finds a point $\left(i',j'\right)$
on the path that is within $\sqrt{N}$ (Manhattan distance) from $\left(i^{*},j^{*}\right)$
requires  querying of $f$ at an expected number of at least
$\Omega\left(\log N\right)$ points on the
path that are pairwise at least $\sqrt{N}$-apart. 
\end{claim}

\begin{proof}
Observe that once we fix the path and the path is known, the values of $f$ outside the
path do not reveal any information about the location of $\left(i^{*},j^{*}\right)$ on the path. 
So, knowing the path, the algorithm only learns from queries on the path. 

Consider an easier problem where for every query $(i,j) \in [N]^2$ the algorithm learns the values of $f$ for every point in the set $\mathcal{Q}(i,j) := \{(i',j') \in [N]^2 \mid i - \sqrt{N} \leq i' \leq i+\sqrt{N} \ \ \mbox{and} \ \  j - \sqrt{N} \leq j' \leq j + \sqrt{N} \}$. 
For this easier problem the queries are clearly $\sqrt{N}$-apart, wlog. 

For the easier problem, Claim~\ref{claim:binary}, in the special case where $p=1$, implies that the algorithm needs $\Omega(\log(\sqrt{N})) = \Omega(\log(N))$ queries in expectation.
The lower bound for the original problem follows because it is harder (weaker queries).
\end{proof}

\subsubsection*{Choosing the central path}

Our goal now is to prove that it is hard to find many distant points
on the path. 
To simplify the analysis, we will only consider the special case where
all points $\left(x,y\right)$ on the path satisfy $x-y\in\left[-N^{1/4},N^{1/4}\right]$.  (Recall that we are assuming wlog that $N^{1/4}$ is an integer.)
We partition the $N \times N$ grid into $\Theta\left(\sqrt{N}\right)$
regions of the form $R_{a}\triangleq\left\{ \left(x,y\right) \in [N]^2 \mid x+y\in[a,a+\sqrt{N})\right\}$, where
$a \in \{ k \sqrt{N} + 2 \mid 0 \leq k \leq (2\sqrt{N}-1) \  \mbox{and} \ k \in \nat \} $.
Notice that each region $R_a$ intersects the path at exactly $\sqrt{N}$
points, except for the last region $R_{2N-\sqrt{N} +2}$ which intersects the path at $\sqrt{N}-1$ points. The path enters each region\footnote{For the first and last region, the path is obviously forced to start
at $\left(1,1\right)$ (respectively end at $\left(N,N\right)$);
but those two regions can only account for two of the $\Omega\left(\log N\right)$
distant path points required by Claim \ref{claim:points-on-path},
so we can safely ignore them.} at a point $\left(x,y\right)$ for an integer value $x-y$ chosen uniformly
at random in the interval $\left[-N^{1/4}, N^{1/4} \right]$. We will argue 
that in order to find a point on the path
in any region $R_{a}$, the algorithm must query the function at $\Omega\left(\log N\right)$
points in $R_{a}$ or its neighboring regions. 

Each region is further partitioned into $\Theta\left(N^{1/4}\right)$
sub-regions $S_{a}\triangleq\left\{ \left(x,y\right)\mid x+y\in[a,a+2N^{1/4})\right\} $.
For each region, we choose a special sub-region uniformly at random.
In all non-special sub-regions, the path proceeds while maintaining
$x-y$ fixed, up to $\pm1$. Inside the special sub-region, the value
of $x-y$ for path points changes from the value chosen at random
for the current region, to the value chosen at random for the next
region. 

Given a choice of random $x-y$ entry point for each region, and a
random special sub-region for each region, we consider an arbitrary
path that satisfies the description above. This completes the description
of the construction.
\begin{claim}[Finding the special sub-region requires $\log(N)$ queries]
\label{claim:Finding-the-special} \hfill

Let $\pi_a$ be the probability that the algorithm queries any point in
region $R_{a}$, and let $q_a$ be the probability that the algorithm queries any point in the {\em special sub-region} in $R_{a}$, conditioning on querying at least one point in $R_{a}$.
Then the algorithm makes $\Omega\left(\pi_a q_a \log (N)\right)$ queries (in expectation) to points in $R_{a}$.
\end{claim}

\begin{proof}
Any query to a point in $R_{a}$ gives no more information about the special sub-region than a comparison query. Thus we can apply Claim~\ref{claim:binary} to get that conditioned on querying at least one point in 
$R_{a}$, the algorithm makes at least $\Omega\left(q_a \log (N)\right)$ queries in expectation. 
\end{proof}

Let $S_{a}$ and $S_{b}$ be the special sub-regions of two consecutive
regions. Let\\ $T\triangleq\left\{ \left(x,y\right)\mid x+y\in[a+2N^{1/4},b)\right\} $
be the union of all the sub-regions between $S_{a}$ and $S_{b}$.
Observe that the value of $x-y$ remains fixed (up to $\pm1$) for
all points in the intersection of the path with $T$. Also, the construction
of $f$ outside $S_{a}\cup T\cup S_{b}$ does not depend at all on
this value. 
\begin{claim}[Finding the path without finding the special sub-region requires $\log(N)$ queries] \hfill

Let $\pi_{a,b}$ be the probability that the algorithm queries any point in
$R_{a} \cup R_{b}$. 
Let $q^{'}
_{a,b}$ denote the probability that the algorithm queries any point in the intersection of the path and $T$ without querying $S_{a} \cup S_{b}$.
Then the algorithm makes $\Omega\left(\pi_{a,b}q^{'}
_{a,b}\log (N)\right)$ queries (in expectation) to points in $R_{a} \cup R_{b}$.
\end{claim}

Notice that in order to query a point in the intersection of
the path and region $R_{a}$, the algorithm must either query the special sub-region of $R_a$ or its neighboring regions, or find the path in between special sub-regions. Summing the probabilities that these events happen, we have by the previous two claims that:
\begin{lemma} [Finding a point on the path requires $\log(N)$ queries]\label{lem:region} \hfill 

Let $\pi_a$ be the probability that the algorithm queries any point in
region $R_{a}$, and let $p_a$ be the probability that the algorithm queries a point in the intersection of the path and region $R_{a}$. Then it makes at least $\Omega\left(\pi_a p_a\log(N)\right)$
queries (in expectation) in $R_{a}$ or its neighboring regions.
\end{lemma}

\subsubsection*{Completing the proof}
Recall that the algorithm finds a point in the intersection of the path with region $R_a$ with probability $\pi_a p_a$.
By Claim~\ref{claim:points-on-path}, we have that 
\begin{gather}\label{eq:points-on-path}
\sum_a \pi_a p_a = \Omega\left( \log(N)\right).
\end{gather}
Therefore the algorithm's total number of queries is at least
\begin{align*}
\frac{1}{3}\sum_a {\text{queries to $R_a$ or neighboring regions}} 
&= \sum_a \pi_a p_a \Omega\left( \log(N)\right) && \text{(Lemma~\ref{lem:region})}\\
& \geq \Omega\left( \log^2(N)\right)  && \text{(Eq.~\eqref{eq:points-on-path}).}
\end{align*}
(Here $\frac{1}{3}$ corrects for the fact that we may triple-count the queries in each region when summing over it and its neighbors.)
%
This completes the
proof of Theorem \ref{thm:main-2d}. \qed

\subsection{An alternative proof}

\begin{theorem}
\label{thm:det-2d}
Any deterministic  black box algorithm for finding a Tarski fixed point in two dimensions needs $\Omega(\log ^2 N)$ queries.
\end{theorem}
This proof could be more promising for generalization: its gist is that any such algorithm must solve
$\Omega(\log N)$ {\em independent one-dimensional problems.}

\begin{proof}

We shall describe a simple strategy for the adversary that achieves
this bound.  The adversary's strategy is to again commit to  
``herringbone'' functions 
as in Figure \ref{fig:2d-f}:
the function consists of a {\em main path} consisting of a
monotonically increasing path from $(1,1)$ to a point $x^*$, and a
monotonically decreasing path from $(N,N)$ to $x^*$, with each
step along the path, except for $x^*$, changing one dimension of the
argument by one unit.  For all points $(x,y)$ off the main path,
$f(x,y)$ is either $(x-1,y+1)$ or $(x+1,y-1)$, depending on whether
$(x,y)$ is below or above the main path, respectively; thus, the graph
of the function is again herringbone-like, consisting of the main path, plus
$45^{\hbox{\rm o}}$ paths towards the main path (see Figure \ref{fig:2d-f}).

For the sake of exposition and geometric intuition, we shall use a
simple notation based on the eight cardinal directions\footnote{We will
try to avoid confusion between the direction {\bf N} (North), and
the number $N$.  That is why we use boldface for the basic directions.}: 
{\bf N}, 
{\bf S}, {\bf E}, {\bf W},
NW, SE, SW, NE. 
Thus, the answer $(x-1,y+1)$ to the query $f(x,y)$
will be denoted NW.  To summarize the adversary's strategy, the answer
to a query $f(x,y)$ is either SE or NW, thus declaring that $(x,y)$ is
not on the path, unless both answers would contradict monotonicity, in
which case the adversary must choose one of the principal directions,
{\bf N}, {\bf S}, {\bf E}, {\bf W}.
A query of the latter type is termed a {\em decisive}
query.  Note that the answer to any non-decisive query $f(x,y)$
effectively ``removes from consideration'' a rectangular area of the
grid --- if $f(x,y)=NW$, the block $\{(x',y'): x'\geq x, y'\leq y\}$,
that is the whole block to the SE of $(x,y)$, is excluded for further
consideration in the sense that the main path can no longer intersect
it, and all points $(x',y')$ in this block must have
$f(x',y')=NW$.\footnote{Strictly speaking point on the block's boundary
  do not have this restriction, but let us assume that they do, as
  this simplification favors the algorithm.}
  At any time, the union of these forbidden rectangles consist 
  of an upper left region that contains all points that are above and/or to 
  the left of the query points $(x,y)$ that point SE (i.e. such that $f(x,y)=SE$) 
  and a lower right region that contains all points that are below 
  or to the right of query  points that point NW.
  The two forbidden regions are bounded by monotone staircase curves,
  and the main path must lie strictly between these two curves.

A query at point $q=(x,y)$ is decisive precisely when both points
  $(x-1,y+1)$ and $(x+1,y-1)$ to the NW and SE of $q$ belong to the forbidden
  area, the first one to (the boundary of) the upper left region and the second
  one to the lower right region. Thus, the main path must pass through the
  query point $q$ and now the adversary must decide whether the fixed point $x^*$ 
  is above or below $(x,y)$.

How is this decision, as well as the decisions off the path (the
choice between NW and SE) made?  At any query, the algorithm has
effectively determined that the part of the main path of current
interest (certain to include the fixed point) is one of the possible
monotonically increasing paths from some point $(\underline
x,\underline y)$ (the SW-most part of the domain), either the origin
or a past decisive query, to some point $(\bar x,\bar y)>(\underline
x,\underline y)$ (the NE-most point of the domain) that avoids all
blocks removed by past non-decisive queries.  We call this region {\em
  the current domain}.  During a decisive query $q$, the algorithm has to
choose: which of the two subdomains of the current domain, the one to
the SW or the one to the NE, will be the new domain?  The answer is
{\em whichever subdomain has the largest number of potential main
  paths.} Since there is at least one potential main path remaining,
at least one of the directions ${\bf S}$, ${\bf W}$ must be available at $q$
(i.e., the point below or to the left of $q=(x,y)$ is not forbidden - it is possible
 that both are available), and similarly at least one of the directions
 ${\bf N}$, ${\bf E}$ must be available at $q$.
  The adversary compares the number of feasible monotone paths in the lower and upper
  subdomain (i.e. the number of feasible monotone paths 
  between $(\underline x,\underline y)$ and $q$,
  to the number between $q$ and $(\bar x,\bar y)$), continues in the subdomain with
  the largest number of paths, and if both choices for direction are available
  in this subdomain, then it chooses again the direction with the larger number of paths.

During any non-decisive query, the same criterion is used: The
adversary will choose the answer among NW and SE that will result in a
new domain (the previous domain with one block removed) with the
largest number of paths that avoid all blocks, among the two possible
choices.  But there is an exception: If the domain is becoming very
narrow --- that is, if the NW or the SE forbidden region is very close
to the query point $q$ - then a different rule is used.  Specifically, if the NW-SE line through the query point $q$ hits the boundary
of the forbidden region on either side within distance $\leq w/2 =N^{\alpha}/2$,
where $\alpha<1$ (for concreteness, assume for
the rest of the proof that $\alpha = 1/2$ and we measure for simplicity the length of diagonal paths in the $L_{\infty}$ metric), then the 
adversary chooses the direction NW or SE from $q$ that is furthest from the
forbidden region (breaking ties arbitrarily). We call such queries {\em short
  queries}.

This completes the description of the adversary's strategy.  The {\em
  potential function} that will inform our lower bound is {\em the
  logarithm of the number of main paths is the current domain}.  That
is, for each time $t$, we define $L_t$ as the logarithm of the number
of monotonically increasing paths in the domain at time $t$ (that is
to say, just before the $t$-th query).  In the beginning, $L_1 \geq N$ ---
actually, it is $2N -{1\over 2}\log N +o(1)$, since the number of
paths is ${2N}\choose{N}$.  When the algorithm concludes, $L_t=0$
(since there is only one path left, the one containing the fixed point).  If the $t$-th query is a decisive
query, then $L_{t+1} \geq {L_t \over 2}-1$, since the number of main
paths before query $t$ was precisely the product of the number
of paths in the upper and lower subdomain, the adversary will choose 
to continue in the subdomain with the largest of the two (thus, with at least
the square root of the number of paths), and if there are two
available choices of direction in the subdomain, it chooses the direction
with the larger number of paths.  

If the $t$-th query $q$ is a non-decisive {\em and non-short} query, then 
all feasible paths, except for those that go through the query point $q$,
belong obviously to either the feasible domain that results if $f(q)=NW$
or the domain that results if $f(q)=SE$.
Since $q$ is not a short query, the number of feasible paths
that go through $q$ is a small fraction of the total number of feasible paths.
Since the adversary chooses the direction among NW, SE with the larger number of paths,
it follows that this number is approximately at least one half of the paths,  
hence certainly $L_{t+1} \geq L_t -2$.  

The following lemma describes what happens at short queries:

\begin{lemma}
If $t$ is a short query, then $L_{t+1} \geq L_t -  N^{\alpha}\log N$.
\end{lemma}
\begin{proof}
Consider a short query $q$ and the NW to SE line through it, which intersects the boundary of the upper left forbidden region at $a$ and the boundary of the lower right region at $b$. Suppose wlog that the adversary in this case chose
$f(q)=NW$, that is, $|qa| \geq |qb|$, where $|qa|, |qb|$ is the length of the
segments $qa, qb$ (in $L_{\infty}$ metric).
Since $q$ is a short query, $d = |qb| \leq w/2$.
Let $s$ be the minimum point of the current domain, and $u$ the maximum.
For a point $p$ of the segment $ab$, we let $n_p$ denote the number of
monotone feasible paths from $s$ to $u$ that go through $p$.
Let $Q$ be the number of paths that go through a point in the $qa$
segment, and $Q'$ the number of paths that go through a point in the $qb$ segment.

Consider a point $p'=(x_{p'},y_{p'}) \in qb$ 
and the point $p=(x_p,y_p) = (x_{p'}-d, y_{p'}+d)$ that is NW of $p'$ 
at distance $d$. 
The point $p$ is in $qa$ since $d=|qb| \leq |qa|$.
Map every $s-u$ path $\pi'$ through $p'$ to the path $\pi$ through $p$,
which agrees with $\pi'$ until it reaches $x-$coordinate $x_p$ for the first time, then $\pi$ moves up vertically to $p$, then horizontally until it meets again $\pi'$, and then follows $\pi'$ until the end (see Figure 2).

\begin{figure}[h]
\centering
\vspace*{-1.2cm}
\includegraphics[scale=0.5]{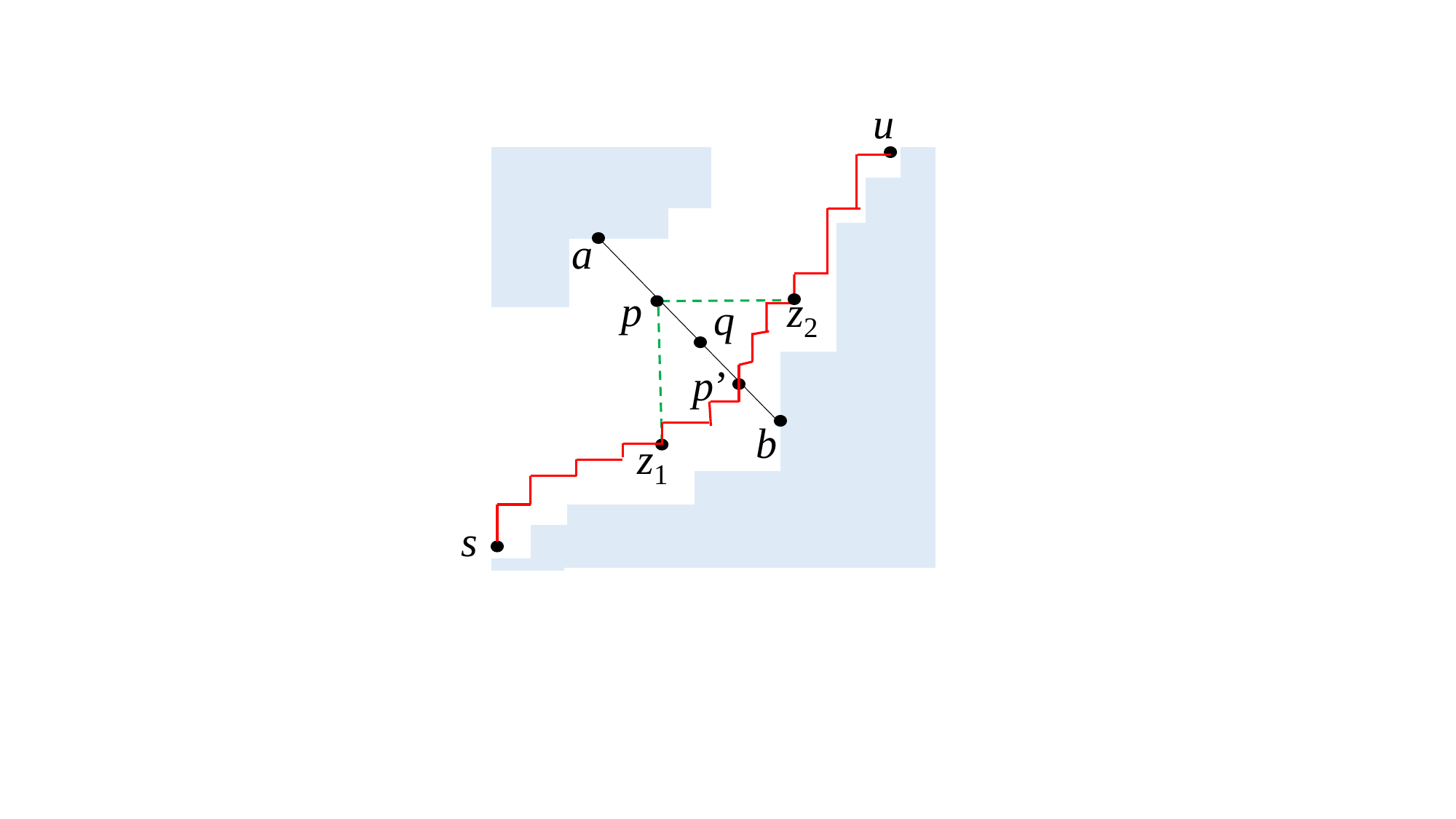}
\vspace*{-3.7cm}
\caption{}
\vspace*{-0.2cm}
\label{fig:short-query}
\end{figure}

Let $z_1$ be the first point of $\pi$ with $x-$coordinate $x_p$, and let
$z_2$ be the last point of $\pi$ with $y$-coordinate $y_p$.
How many paths $\pi'$ through $p'$ get mapped to the same path $\pi$ through $p$?
All these paths $\pi'$ differ only in their portion between $z_1$ and $z_2$.
The number of monotone paths from $z_1$ to $p'$ is at most $N \choose d$,
because such a path amounts to choosing $d$ {\bf E} moves out of at most $N$ steps
(and some of these paths may in fact not be feasible),
and similarly the number of monotone paths from $p'$ to $z_2$ is at most
$N \choose d$.
Therefore, $n_{p'} \leq {N \choose d}^2 \cdot n_p$, for every $p' \in qb$,
and consequently $Q' \leq {N \choose d}^2 \cdot Q$.
We have to account also for the $s-u$ paths through the query point $q$.
If $|qa| > |qb|$, then we can map them to the paths 
through the point at distance $d$
NW of $q$, but even if $|qa|=|qb|$, note similarly that the number of paths through
$q$ is at most $N^2$ times the number of paths through the point immediately NW of it.
In any case, since $d \leq w/2$, the total number of paths before the
$t$-th query is 
$2^{L_t} \leq 2 {N \choose {w/2}}^2 \cdot Q < N^w \cdot Q =N^w \cdot 2^{L_{t+1}}$.
The lemma follows.
\end{proof}

The rest of the lower bound argument proceeds as follows: We shall
show that there are at least $\Theta(\log N)$ decisive queries such
that we can ``charge'' to each of them $\Theta(\log N)$ other queries
--- naturally, a query should be charged to only one decisive query,
or at most a constant number of them.  
The theorem then follows immediately.  The first part, the
existence of $\Theta(\log N)$ decisive queries, is already obvious;
the $\Theta(\log N)$ queries that can be charged to each (without
much overcharging) will take a little more care to establish.
We show first that there is a set of $\Omega(\log N)$ decisive queries
that are $w$ far from each other in both coordinates.

\begin{lemma} \label{lem:effective}
If the total number of queries is no more than $log ^2 N$, then there is a set $S$ of $K=\Omega(\log N)$ decisive queries 
$\{q_1=(x_1,y_1),\ldots,q_K=(x_K,y_K)\}$ such that, for any $1\leq i\neq j\leq K$ we have that $|x_i-x_j|,|y_i-y_j|> w$.
\end{lemma}
\begin{proof}
Any decisive query $q_t=(x,y)$ takes place within a domain $D_t$ with
a SW-most point $(\underline{x},\underline{y})$ and a NE-most point
$(\bar{x},\bar{y})$.  We claim that, if $x$ is within $w$ of
$\underline{x}, \bar{x}$, or if $y$ is within $w$ of $\underline{y},
\bar{y}$, then this query decreases $L_t$ by at most $w\log N$. In
proof, if $x$ is within $w$ of $\bar{x}$ the number of paths between
$(x,y)$ and $(\bar{x},\bar{y})$ is at most ${N\choose w}\leq N^w$, and
a similar argument holds for the other directions.  Let us call such
decisive queries {\em ineffective}, and otherwise they are called {\em effective}.

In summary, we have a potential function that starts at the value $N$, and then is decreased in at most $\log^2 N$ steps 
either (1) by a factor of no more than two, minus additive 1
(decisive queries that are effective), or (2) by an additive term of at most 
$N^\alpha\log N$ (non-decisive, or ineffective decisive queries).  
It follows from arithmetic that there must be at least 
$\log\left({N^{1-\alpha}\over {\log^3 N}}\right)$ queries of type (1).

Hence there is a set $S$ of $K=\Omega(\log N)$ decisive queries that
are effective.  At the time each of these queries was issued, it
was farther than $w$ from the SW and NE corners of its domain, in both
the $x$ and the $y$ direction, and thus also farther than $w$ from any
other previous queries --- and this includes the previous decisive queries in
$S$.  Hence these queries are all farther than $w$ away from each
other, as claimed in the lemma.
\end{proof}

We show now how to assign $\Omega(\log n)$ non-decisive queries to each
`effective' query $q$ in the set $S$ of the previous lemma.
These are essentially the trace of the binary search that helped 
the algorithm corner the adversary into $q$.
We will refer in the following to the two boundary half-lines of the forbidden block
generated by a non-decisive query as its {\it walls}.
Consider a decisive effective query $q=(x,y)$ at time $t$.

\begin{lemma} \label{lem:walls}
For every decisive effective query point $q$ 
there are $\Omega(\log N)$ walls, generated by non-decisive queries,
 that intersect the NW-SE line through $q$ 
within a distance $w/2$ from $q$.
\end{lemma}
\begin{proof}
Since the query point $q=(x,y)$ is decisive,
the points $p_1=(x-1,y+1)$, $p_2=(x+1,y-1)$ that are at distance 1 NW and SE from $q$
belong to the forbidden region, hence  there exist two walls
within a distance of 1 from $q$ on the NW-SE line, on either side of $q$, corresponding
to two queries $q_1=(x_1,y_1)$ and $q_2=(x_2,y_2)$, at times
$t_1, t_2$.
Since $q$ is effective, the queries $q_1, q_2$ are non-decisive.
We will use induction on $k =2, \ldots, \lfloor \log(w/2) \rfloor$, to show that
there is a set $S_k$ of  $k$ walls, generated by
non-decisive queries, that intersect the NW-SE line through $q$ on both sides,
within an interval that includes the point $q$ and has length $\delta_k \leq 2^k$
(in the $L_{\infty}$ metric).
This claim for $k=\lfloor \log(w/2) \rfloor$ implies immediately the lemma.
For the basis, $k=2$, we let $S_2$ contain the walls at $p_1$ and $p_2$.

For the induction step, consider the set $S_k$ of walls.
Let $t_l$ be the earliest time that generated a wall of $S_k$ that intersects 
the NW-SE line through $q$ left of $q$ (i.e. NW of $q$), let $p_l$ be the
intersection point and $q_l$ the query point that generated the wall.
Similarly, let $t_r$ be the earliest time that generated a wall of $S_k$ that intersects 
the NW-SE line through $q$ right of $q$ (i.e. SE of $q$), let $p_r$ be the
intersection point and $q_r$ the query point that generated the wall.

Suppose without loss of generality that $t_l > t_r$. 
Why did the adversary choose SE in response to query $q_l$ at time $t_l$? 
Since $t_l > t_r$, the walls of $q_r$ existed at time $t_l$.
The wall through $p_l$ is either vertical, in which case $q_l$ is below $p_l$,
or the wall is horizontal, in which case $q_l$ is to the right of $p_l$
(see Figure 3).
In either case, it is easy to see that the line from $q_l$ in the SE direction
hits a wall of the query point $q_r$ within distance at most the length $|p_l p_r |$ of the
segment $(p_l,p_r)$, thus at most $2^k$; 
Fig. 3 shows the geometry when the wall at $p_r$ is vertical (the case
of a horizontal wall at $p_r$ is symmetric).

\begin{figure}[h]
\centering
\vspace*{-1.1cm}
\includegraphics[scale=0.5]{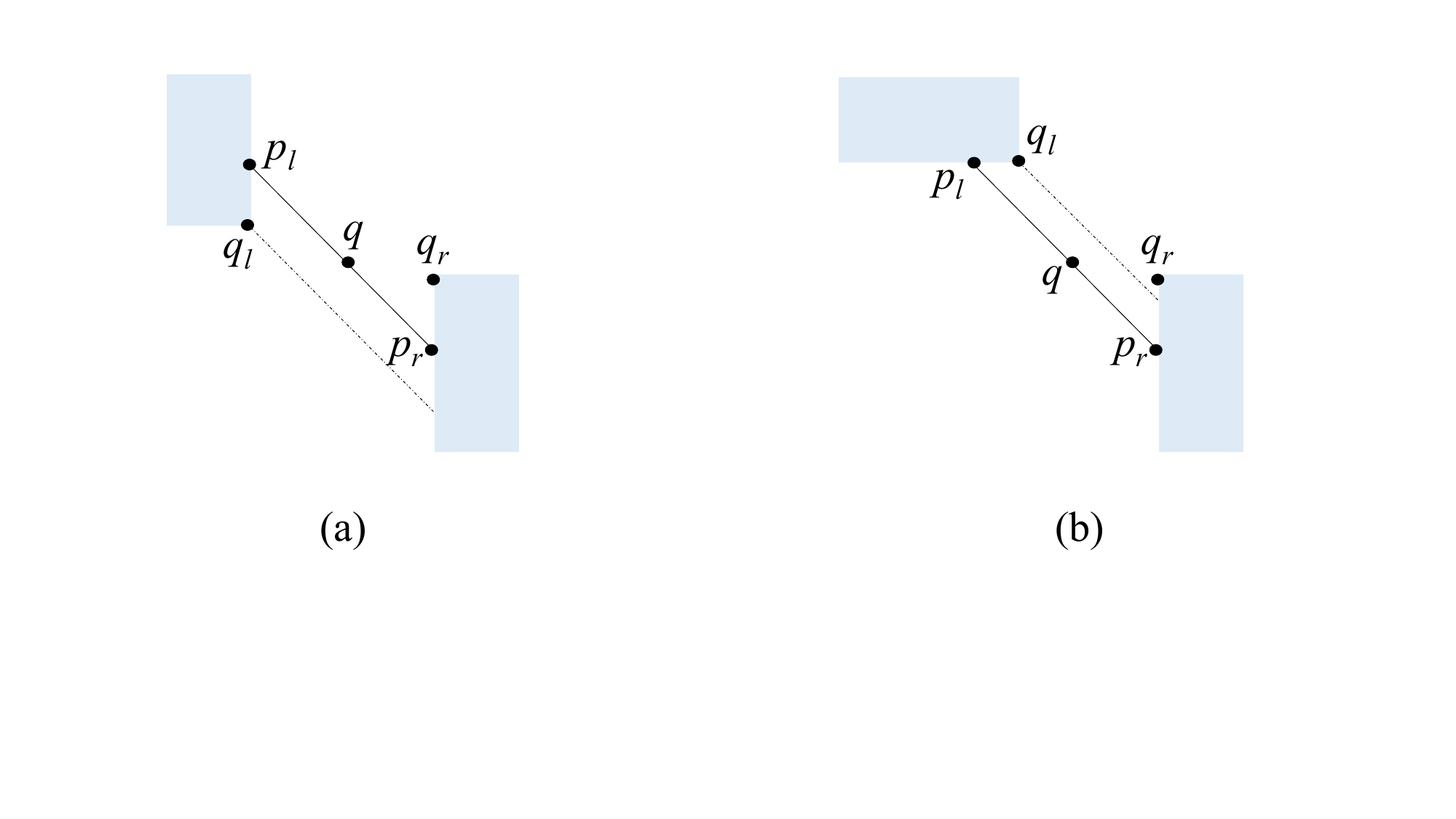}
\vspace*{-4cm}
\caption{}
\vspace*{-0.2cm}
\label{fig:wall}
\end{figure}

Since the line from $q_l$ in the SE direction hits a wall within $2^k  \leq w/2$, 
$q$ is a short query.
Since the adversary chooses SE at $q_l$, the line from $q_l$ in the NW direction
must hit also within distance at most $2^k$ another wall, generated
by a query point $q_{k+1}$ at an earlier time $t_{k+1} < t_l$.
Since $q_l$ is below or to the right of $p_l$, the NW-SE line through $q$
hits a wall of $q_{k+1}$ at a point $p_{k+1}$ that is at most $2^k$ beyond $p_l$.
Adding this wall to $S_k$ yields the set $S_{k+1}$ that satisfies the
induction hypothesis.
\end{proof}

We can now complete the proof of Theorem \ref{thm:det-2d}.
By Lemma \ref{lem:walls}, to every effective decisive query $q$ we can assign
$\Omega(\log N)$ non-decisive queries that generate walls within
$w/2$ of $q$, hence their $x-$ or $y-$coordinate is within $w/2$
of that of $q$. Since the $\Omega(\log N)$ effective queries of the set $S$ of Lemma \ref{lem:effective} are more than $w$ far from each other in both coordinates,
a non-decisive query can be close to at most one query of $S$
in $x$-coordinate and at most one in $y$-coordinate.
Therefore, there are $\Omega(\log^2 N)$ distinct non-decisive queries.
\end{proof}

\section{Supermodular Games}

\label{sec:super-modular}

\subsection{A brief intro to supermodular games}

A {\em supermodular game} is a game in which the set $S_i$ of strategies
of each player $i$ is a complete lattice, and the utility (payoff) functions $u_i$ satisfy certain conditions.
Let $k$ be the number of players and let $S= \Pi_{i=1}^k S_i$ be the set of strategy profiles.
As usual, we use $s_i$ to denote a strategy for player $i$ and
$s_{-i}$ to denote a tuple of strategies for the other players.
The conditions on the utility functions $u_i$ are the following:\\
C1. $u_i(s_i, s_{-i})$ is upper semicontinuous in $s_i$ for fixed $s_{-i}$, and it is continuous in $s_{-i}$ for each fixed $s_i$, and has a finite upper bound.\\
C2. $u_i(s_i, s_{-i})$ is supermodular in $s_i$ for fixed $s_{-i}$.\\
C3. $u_i(s_i, s_{-i})$ has increasing differences in $s_i$ and $s_{-i}$.

A function $f: L \rightarrow \real$ is {\em supermodular} if for all
$x,y \in L$, it holds $f(x)+f(y) \leq f(x \land y)+ f(x \lor y)$.
A function $f: L_1 \times L_2 \rightarrow \real$, where $L_1, L_2$ are lattices, has {\em increasing differences} in its two arguments,
if for all $x' \geq x$ in $L_1$ and all $y' \geq y$ in $L_2$, it holds that $f(x',y')-f(x,y') \geq f(x',y) -f(x,y)$.

The broader class of {\em games with strategic complementarities} (GSC)
relaxes somewhat the conditions C2 and C3 into C2', C3' which depend only
on ordinal information on the utility functions, i.e. how the utilities compare to each other rather than their precise numerical values. 
The supermodularity requirement of C2
is relaxed to quasi-supermodularity, where a function $f: L \rightarrow \real$ is {\em quasi-supermodular} if for all
$x,y \in L$, $f(x) \geq f(x \land y)$ implies $f(y) \leq f(x \lor y)$,
and if the first inequality is strict, then so is the second.
The increasing differences requirement of C3 is relaxed to the {\em single-crossing condition}, where a function $f: L_1 \times L_2 \rightarrow \real$, satisfies the
single crossing condition,
if for all $x' > x$ in $L_1$ and all $y' > y$ in $L_2$, it holds that $f(x',y) \geq f(x,y)$ implies  $f(x',y') \geq f(x,y')$, and if the first inequality is strict then so is the second.
All the structural and algorithmic properties below of supermodular games
hold also for games with strategic complementarities.
 
We will consider here games where each $S_i$ is a discrete (or continuous) finite box in $d_i$ dimensions of size $N$ in each coordinate. We let $d=\sum_{i=1}^k d_i$ be the total number
of coordinates.
In the discrete case, condition C1 is trivial.
Condition C2 is trivial if $d_i=1$ 
(all functions in one dimension are supermodular), 
but nontrivial for 2 or more dimensions.
C3 is nontrivial.

Supermodular games (and GSC) have pure Nash equilibria. 
Furthermore, the pure Nash equilibria form a complete lattice \cite{MR90},
thus there is a highest and a lowest equilibrium.
Another important property is that the best response correspondence 
$\beta_i(s_{-i})$ for each player $i$ has the property that
(1) both $\sup \beta_i(s_{-i})$ and $\inf \beta_i(s_{-i})$ are in
$\beta_i(s_{-i})$, and (2) both functions $\sup \beta_i (\cdot)$
and $\inf \beta_i(\cdot)$ are monotone functions \cite{Topkis98}.
The function $\bar{\beta}(s) = ( \sup \beta_1(s_{-1}), \ldots , \sup \beta_k(s_{-k}) )$ of the supremum best responses is a monotone function from $S$ to itself, and its greatest fixed point is the highest Nash equilibrium of the game.
The function $\underline{\beta}(s) = ( \inf \beta_1(s_{-1}), \ldots , \inf \beta_k(s_{-k}) )$ of the infimum best responses is also a monotone function, and its least fixed point is the lowest Nash equilibrium of the game.

\begin{example}{(A simplified Diamond search model \cite{MR90}.)} \ 
{\rm There are $k$ players (businesses).  Each player $i \in [k]$ can exert some amount of ``effort'', $s_i \in [0,m_i]$,  where $m_i >0$, to find
a business partner.  So, the strategy space $S_i$ of player $i$ is the closed bounded interval $[0,m_i]$. 
Any player $i$ incurs a cost $C_i(s_i)$ for exerting effort $s_i$, where we assume $C_i(\cdot)$ is some arbitrary continuous function
(we do not necessarily assume  that $C_i(s_i)$ is increasing in $s_i$; this is not needed).  
The payoff to player $i$ depends also on how much effort others are putting into finding a business partner.
Specifically, for each player $i$, we assume that for some $\alpha_i >0$ the utility function $u_i(s_1,\ldots,s_k)$ for player $i$ is given by:
$$u_i(s) :=  \alpha_i \cdot s_i \cdot (\sum_{j \neq i} s_j)  -  C_i(s_i)$$
Let us check that this is a supermodular game.  Clearly the strategy space $S_i = [0,m_i]$ of each player (a closed interval) is a complete lattice. \\   
C1. condition C1 certainly holds, since in fact $u_i(s_i,s_{-i})$ is continuous
in both $s_i$ and $s_{-i}$, and has a finite upper bound (because the strategy spaces of all players are bounded intervals).\\
C2. condition C2 holds vacuously, because for fixed $s'_{-i}$, the function $f(s_i) := u_i(s_i, s'_{-i})$
is a function in a single real-valued parameter, $s_i$, and 
any such function is supermodular, because  for all $s_i, s'_i \in S_i$,  $f(s_i) + f(s'_i) = f(\min \{s_i, s'_i\}) + f(\max\{s_i,s'_i\}) = 
f(s_i \land s'_i) + f(s_i \lor s'_i)$.\\
C3. To see that condition C3 holds,  i.e., that the payoff functions $u_i(s_i,s_{-i})$ have increasing
differences in $s_i$ and $s_{-i}$, 
suppose that $s'_i \geq s_i$  and $s'_{-i} \geq s_{-i}$  (coordinate-wise inequality).
Then note that we have:
\begin{eqnarray*}
u_i(s'_i,s'_{-i}) - u_i(s_i,s'_{-i}) & = &   \alpha_i s'_i (\sum_{j \neq i} s'_j) - C_i(s'_i)   -  ( \alpha_i s_i (\sum_{j \neq i} s'_j) - C_i(s_i))\\
                                     & = & \alpha_i (s'_i - s_i) (\sum_{j \neq i} s'_j) - C_i(s'_i) + C_i(s_i) \\
                                     & \geq & \alpha_i (s'_i - s_i) (\sum_{j \neq i} s_j) - C_i(s'_i) + C_i(s_i)\\
                                     & = &  u_i(s'_i,s_{-i}) - u_i(s_i,s_{-i}) 
\end{eqnarray*}}
\qed
\end{example}

\subsection{Complexity of equilibrium computation in supermodular games}

Given a supermodular game, the relevant problems include: (a) find a Nash equilibrium (anyone)\footnote{Whenever we speak
of finding a Nash Equilibrium (NE) for a supermodular game, we mean a {\em pure} NE, as we know that
these exists.}, and (b) find the highest or the lowest equilibrium. In the case of continuous domains, we again have to relax to an approximate solution.
We assume that we have access to a best response function, e.g. $\bar{\beta}(\cdot)$ and/or $\underline{\beta}(\cdot)$, as an oracle or as a polynomial-time function.
The monotonicity of these functions implies then easily the following:

\begin{proposition}
\label{prop:supmod-to-tarski-reduct}
1. The problem of computing a pure Nash equilibrium of a $k$-player supermodular game over a discrete finite strategy space $\Pi_{i=1}^k [N]^{d_i}$ reduces to the problem of computing a fixed point of a monotone function over $[N]^d$ where $d=\sum_{i=1}^k d_i$. 
Computing the highest (or lowest) Nash equilibrium reduces to computing the greatest (or lowest) fixed point of a monotone function.

2. For games with continuous box strategy spaces, $\Pi_{i=1}^k [1,N]^{d_i}$, and Lipschitz continuous utility functions with Lipschitz constant $K$, 
the problem of computing an
$\epsilon$-approximate pure Nash equilibrium reduces to 
exact fixed point computation for a monotone function with a discrete finite domain $[NK/\epsilon]^d$.
\end{proposition}
\begin{proof}
1. Follows from the monotonicity of $\bar{\beta}(\cdot)$ and $\underline{\beta}(\cdot)$. If $s$ is fixed point of $\bar{\beta}(\cdot)$, then $s_i =\sup \beta_i(s_{-i})$ is a best response to $s_{-i}$
 for all $i$ (since $\sup \beta_i(s_{-i}) \in \beta_i(s_{-i})$),
therefore $s$ is a Nash equilibrium of the game.
The GFP of $\bar{\beta}(\cdot)$ is the highest Nash equilibrium.
Similarly, every fixed point of $\underline{\beta}(\cdot)$ is
an equilibrium of the game, and the LFP of $\underline{\beta}(\cdot)$
is the lowest equilibrium.

2. Suppose that the utility functions are Lipschitz continuous with
Lipschitz constant $K$. To compute an $\epsilon$-approximate Nash equilibrium of the game, it suffices to find a $\epsilon/K$-approximate
fixed point of the function $\bar{\beta}(\cdot)$.
For, if $s$ is such an approximate fixed point and
$s'=\bar{\beta}(s)$, then
$|s' - s| \leq \epsilon/K$ in every coordinate.
Hence $|u_i(s_i,s_{-i})-u_i(s'_i,s_{-i})| \leq \epsilon$,
and $s'_i$ is a best response to $s_{-i}$, hence $s$ is
an $\epsilon$-approximate equilibrium.
Computing an $\epsilon/K$-approximate
fixed point of the function $\bar{\beta}(\cdot)$ on the continuous domain, reduces by Proposition \ref{con-disc} 
to the exact fixed point problem for
the discrete domain $[NK/\epsilon]^d$.
\end{proof}

Not every monotone function can be the (sup or inf) best response function of a game.
In particular, a best response function has the property
that the output values for the components corresponding to a player depend only
on the input values for the other components corresponding to the other players.
Thus, for example, for two one-dimensional players, if the function $f(x,y)$ 
is the best response function of a game, it must satisfy 
$f_1(x,y)=f_1(x',y)$ for all $x,x',y$, and $f_2(x,y)=f_2(x,y')$ for all $x,y,y'$.
This property helps somewhat in improving the time needed to find a fixed point,
and thus an equilibrium of the game, as noted below. 
For example, in the case of two one-dimensional players, an equilibrium
can be computed in $O(\log N)$ time, instead of the $\Omega (\log^2 N)$ time 
needed to find a fixed point of a general monotone function in two dimensions.

\begin{theorem}\label{sup:upper}
Given a supermodular game with two players with discrete strategy spaces
$[N]^{d_i}$, $i=1,2$ with access to the sup (or inf) best response function
$\bar{\beta}(\cdot)$ (or $\underline{\beta}(\cdot)$),
we can compute an equilibrium in time $O( (\log N)^{\min(d_1,d_2)} )$.
More generally, for $k$ players with dimensions $d_1, \ldots , d_k$,
an equilibrium can be computed in time $O( (\log N)^{d'} )$, 
where $d' = \sum_i d_i - \max_i d_i$.
\end{theorem}
\begin{proof}
Suppose that we have access to the sup best response $\bar{\beta}(\cdot)$.
Assume without loss of generality that the first player has the maximum
dimension, $d_1 = \max_i d_i$. 
We apply the divide-and-conquer algorithm, but take advantage of the 
property of the monotone function $\bar{\beta}$ that the first $d_1$
components of $\bar{\beta}(x)$ do not depend on the first $d_1$ coordinates of $x$.
As a consequence, for any fixed assignment to the other coordinates,
i.e. choice of a strategy profile $s_{-1}$ for all the players except the first player,
the induced function on the first $d_1$ coordinates maps every point
to the best response $\bar{\beta}_1 (s_{-1})$ of player 1.
Thus the fixed point of the induced function is simply $\bar{\beta}_1 (s_{-1})$,
it can be computed with one call to $\bar{\beta}$, 
and there is no need to recurse on the first $d_1$ coordinates.
It follows that the algorithm takes time at most $O( (\log N)^{d'} )$, 
where $d' = \sum_i d_i - \max_i d_i$.
\end{proof}

Conversely, we can reduce the fixed point computation problem for an arbitrary
monotone function to the equilibrium computation problem for a supermodular game
with two players.

\begin{theorem}\label{mon-to-sup}
1. Given a monotone function $f$ on $[N]^d$ (resp. $[1,N]^d$) 
we can construct a supermodular game
$G$ with two players, each with strategy space $[N]^d$ (resp. $[1,N]^d$), so that the equilibria of $G$ correspond to the fixed points of $f$.

2. More generally, the fixed point problem for a monotone function $f$ in $d$ dimensions can
be reduced to the equilibrium problem for a supermodular game with any number $k \geq 2$ of players with any dimensions $d_1, \ldots, d_k$, provided that $\sum_i d_i \geq 2d$
and $\sum_i d_i - \max_i d_i \geq d$.
\end{theorem}
\begin{proof}
1. We will define the utility functions $u_i$ so that the best responses $\beta_i$ of both players
are functions (i.e. are unique). For player 1, the best response will be 
$\beta_1(y)=y$, for all $y \in [N]^d$, and for player 2, the best response will be 
$\beta_2(x)=f(x)$, for all $x \in [N]^d$.
If $x$ is a fixed point of $f$, then $(x,x)$ is an equilibrium of the game,
since $\beta(x,x) =(x,f(x))=(x,x)$.
Conversely, if $(x,y)$ is an equilibrium of the game, then
$\beta(x,y)=(x,y)$, therefore $x=y$ and $y=f(x)$, hence $x=f(x)$.
Thus, the set of equilibria of $G$ is $\{(x,x) | x \in Fix(f) \}$.

The utility function for player 1 is set to 
$u_1(x,y) = -(x-y)^2 = -\sum_{j=1}^d (x_j-y_j)^2$.
The utility function for player 2 is 
$u_2(x,y) = -(f(x)-y)^2 = -\sum_{j=1}^d (f_j(x)-y_j)^2$.
Obviously, the best response functions are as stated above,
$\beta_1(y)=y$ and $\beta_2(x)=f(x)$.

The utility functions $u_1, u_2$ satisfy condition C2 with equality.
For example, to check $u_2$ ($u_1$ is similar), 
fix a $x$ and consider two values $y, y'$.
For every $j=1,\ldots, d$, we have
$-(f_j(x)-y_j)^2 - (f_j(x)-y'_j)^2 = -(f_j(x)-\max (y_j,y'_j))^2 - (f_j(x)- \min (y_j,y'_j))^2$.
Summing over all $j$ yields:
$-(f(x)-y)^2 - (f(x)-y')^2 = -(f(x)- \max(y,y'))^2 - (f_j(x)- \min(y,y'))^2$.

To verify condition C3 for $u_2$, consider any $x' \geq x$ and $y' \geq y$. 
We have $u_2(x',y')-u_2(x,y') - (u_2(x',y) -u_2(x,y))$
$= -\sum_{j=1}^d (f_j(x')-y'_j)^2 +\sum_{j=1}^d (f_j(x)-y'_j)^2
+\sum_{j=1}^d (f_j(x')-y_j)^2 -\sum_{j=1}^d (f_j(x)-y_j)^2$
$= \sum_{j=1}^d 2(y'_j-y_j)(f_j(x')-f_j(x)) \geq 0$,
where the last inequality holds because $y' \geq y$
and $f(x') \geq f(x)$ since $x' \geq x$ and $f$ is monotone.
Similarly, condition C3 can be verified for $u_1$.

2. Order the players in increasing order of their dimension,
let $T$ be the ordering of all the $\sum_{i=1}^k d_i$ coordinates consisting
first of the set $Co(1)$ of coordinates of player 1 (in any order), then the
set $Co(2)$ of coordinates of player 2, and so forth.
Number the coordinates in the order $T$ from 1 to $\sum_{i=1}^k d_i$,
and label them cyclically with the labels $1, \ldots, d$.

We define the (unique) best response function $\beta$ as follows.
For every coordinate $j \leq d$ (in the ordering $T$), 
we set $\beta_j (x)=f_j(x')$, where $x'$ is a subvector of $x$ with $d$ coordinates
that have distinct labels $1,\ldots,d$ and which belong to different players than coordinate $j$.
The subvector $x'$ is defined as follows. Suppose that coordinate $j$ belongs to player $r$
($j \in Co(r)$),
and let $t=\sum_{i=1}^{r-1} d_i$. If $d_r \leq d$, then $x'$ is the subvector of $x$
that consists of the first $t$ coordinates (in the order $T$) and the coordinates $t+1+d, \ldots, 2d$;
note that all these coordinates do not belong to player $r$.
If $d_r >d$, then $r<k$ (since $\sum_i d_i - \max_i d_i \geq d$). In this case, let
$x'$ be the subvector of $x$ consisting of the last $d$ coordinates (in $T$); all of these
belong to player $k \neq r$.
For coordinates $j >d$, we set $\beta_j (x)= x_{j'}$, where
$j' \in [d]$ is equal to $j \mod d$, unless $j'$ belongs to the same player $r$ as $j$,
in which case $d_r >d$, hence $r \neq k$;  in this case we set $\beta_j (x)= x_{j"}$
for some (any) coordinate $j"$ of the last player $k$ that is labeled $j'$.

We define the utility functions of the players so that they yield the
above best response function $\beta$.
Namely, we define the utility function of player $i$ to be 
$u_i(x) =-\sum_{j \in Co(i)} (x_j - \beta_j (x))^2$.
It can be verified as in part 1 that the utility functions 
satisfy conditions C2 and C3.
It can be easily seen also that at any equilibrium of the game,
all coordinates with the same label must have the same value,
and the corresponding $d$-vector $x$ is a fixed point of $f$.
Conversely, for any fixed point $x$ of $f$,
the corresponding strategy profile of the game is an equilibrium.
\end{proof}

Since the 2-dimensional monotone fixed point problem requires $\Omega (\log^2 N)$ queries by
Theorem \ref{thm:main-2d}, it follows that the equilibrium problem for two 2-dimensional players
also requires $\Omega (\log^2 N)$ queries, which is tight
because it can be also solved in $O (\log^2 N)$ time by
Theorem \ref{sup:upper}.
Similarly, for higher dimensions $d$, if the monotone fixed point problem
requires $\Omega (\log^d N)$ queries then the equilibrium problem for two $d$-dimensional players is also $\Theta (\log^d N)$. 

The same reduction from monotone functions to supermodular games of Theorem \ref{mon-to-sup},
combined with Proposition \ref{lfp:hard} implies the hardness
of computing the highest and lowest equilibrium.

\begin{corollary}
\label{cor:l-g-equil-np-hard}
It is \NP-hard to compute the highest and lowest equilibrium of a supermodular game
with two 1-dimensional players with explicitly given polynomial-time best response 
(and utility) functions.
\end{corollary}

\section{Condon's and Shapley's stochastic games  
reduce to \Tarski{}}

\label{sec:condon-shapley}

In this section we show that computing 
the {\em exact} (rational) value of Condon's
simple stochastic games (\cite{Condon92}), as well as computing
the (irrational) value
of Shapley's more general (stopping/discounted) stochastic games \cite{Shapley53}
to within a given desired error $\epsilon > 0$ (given in binary),
are both polynomial time reducible to
$\Tarski{}$.

\subsection{Condon's simple stochastic games reduce to \Tarski{}}

Recall
that a {\em simple stochastic game}\footnote{The definition we give here 
for SSGs is slightly
  more general than Condon's original definition in \cite{Condon92}.  
Specifically, Condon allows edge probabilities of $1/2$ only, and 
also assumed that the game is a ``stopping game'', meaning it halts with probability 1, regardless
of the strategies of the two players.
It is well known that our more general definition does not alter  the
  difficulty of computing the game value and optimal
strategies: solving general SSGs can be reduced in P-time to solving SSGs 
in Condon's more restricted form.  (This follows, e.g., by an easy adaptation of the proof of Lemma 8 in \cite{Condon92}.)}
(SSG) is a 2-player zero-sum game,
played on the vertices of an edge-labeled directed graph,
specified by
$G = (V,V_0,V_1,V_2,\delta)$,
whose vertices $V = \{v_1,\ldots,v_n\}$ include two
special sink vertices,  a  ${\mathbf{0}}$-sink, $v_{n-1}$, and a 
$\mathbf{1}$-sink, $v_{n}$, and
where
the rest of the vertices $V \setminus \{ v_{n-1}, v_{n}\} = 
\{v_1,\ldots,v_{n-2}\}$ 
are partitioned into three disjoint sets $V_0$ (random), $V_1$ (max), and $V_2$ (min).
The labeled directed edge relation is 
$\delta \subseteq (V \setminus \{v_{n-1} , v_n \}) 
\times ((0,1] \cup \bot) \times
V$.
For each ``random'' node $u \in V_0$, 
every outgoing edge $(u,p_{u,v},v) \in \delta$ is labeled by a positive 
probability
$p_{u,v} \in (0,1]$,  such that these probabilities sum to $1$, i.e., $\sum_{
\{v \in V \mid (u,p_{u,v},v) \in \delta\}} p_{u,v} = 1$.
We assume, for computational purposes, that the probabilities $p_{u,v}$
are rational numbers (given as part of the input, with numerator and denominator given in binary).
The outgoing edges from ``max'' ($V_1$) and ``min'' ($V_2$) nodes 
have an empty label, ``$\bot$''.  
We assume each vertex $u \in V \setminus \{ v_{n-1} , v_{n} \}$
has at least one outgoing edge.
Thus in particular, for any node $u \in V_1 \cup V_2$ 
there exists an outgoing edge $(u,\bot,v) \in \delta$ for some $v \in V$. 
Finally, there is a designated start vertex $s \in V$.

A play of the game transpires as follows:  a token is initially
placed on $s$, the start node. Thereafter, during
each ``turn'',  when the token is currently on a node $u \in V$,
unless $u$ is already a sink node (in which case the game halts), the token is moved across
an outgoing edge of $u$ to the next node by whoever ``controls'' $u$. 
For a random node $u \in V_0$, which 
is controlled by ``nature'', the outgoing
edge is chosen randomly according to the probabilities $(p_{u,v})_{v \in V}$.
For $u \in V_1$, the outgoing edge is chosen  by
player 1, the $\max$ player, who aims to maximize the
probability that the token will eventually reach the $\mathbf{1}$-sink.
For $u \in V_2$, the outgoing edge is chosen by player 2,
the $\min$ player, who aims to minimize the probability
that the token will eventually reach the $\mathbf{1}$-sink.
The game halts if the token ever reaches either of the two sink nodes.

For every possible start node $s = v_i \in V$, this zero-sum game has a well 
defined {\em value},
$q^*_i \in [0,1]$.  This is, by definition,
a probability such that player 1, the max player  (and, respectively, 
player 2, the 
min player)
has a strategy to ``force'' 
reaching the $1$-sink  with probability at least (respectively, at most)
$q^*_i$, 
irrespective of what the other player's strategy is.
In other words, these games are {\em determined}. 
Moreover 
$q^*_i$ is a
rational value whose encoding size, with numerator and denominator in binary, 
is polynomial in the bit encoding size of the SSG
(\cite{Condon92}).  Furthermore, both players have  deterministic,  
memoryless (a.k.a., pure, positional) optimal strategies in the game
(which do not depend on the specific start node $s$), in which for each vertex $u \in V_1$ (or $u \in V_2$)
the max player (respectively the min player) chooses the same specific outgoing edge
every time the token visits vertex $u$, 
regardless of the prior history of play prior to that visit to $u$.

Given an SSG, the goal is to compute the value of the game (starting at each vertex).
Condon (\cite{Condon92})
already showed that the problem of deciding whether the value
is $> 1/2$ is in \NP{} $\cap$ \coNP, and it is a long-standing open problem
whether this is in P-time.
Moreover, the search problem of computing the value for an SSG is known to be in both $\PLS{}$ and $\PPAD{}$
(see, e.g., \cite{Yan90} and \cite{EY2010}).

\begin{proposition}
The following total search problem is polynomial-time reducible to $\Tarski{}$: 
Given an instance $G$ of Condon's simple stochastic game, and given
a start vertex $s = v_i \in V$, 
compute the exact (rational) value $q^*_i$ of the game.
\end{proposition}

\begin{proof}
Let $x = (x_1,\ldots,x_n)$ be an $n$-vector of variables.
The $n$-vector $q^*$ of values, $q^*_i$, of the
SSG starting at each vertex $v_i$, 
is given by the {\em least fixed point} (LFP) solution 
of the following monotone min-max-linear system 
of $n$ equations in $n$ unknowns:

\[   x_i =  \left\{ \begin{array}{ll}
                        
\sum_{ \{ v_j \in V \mid (v_i,p_{v_i,v_j},v_j) \in \delta\} } p_{v_i,v_j} x_j &  
\mbox{if $v_i \in V_0$}\\
\max \{ x_j \mid  (v_i, \bot, v_j) \in \delta\}   &  \mbox{if $v_i \in V_1$}\\
\min \{ x_j \mid (v_i,\bot, v_j) \in \delta \}  &  \mbox{if $v_i \in V_2$}\\
0  &   \mbox{if $v_i = v_{n-1}$ is the $\mathbf{0}$-sink}\\
1  &    \mbox{if $v_i = v_{n}$ is the $\mathbf{1}$-sink}
\end{array}
\right. \]

We denote this system of equations by $x=F(x)$. 
Note that $F(x)$ defines a monotone map $F:[0,1]^n \rightarrow [0,1]^n$
mapping the complete lattice 
$[0,1]^n$ (under coordinate-wise order) to itself.  Thus by 
Tarski's theorem it has a least (as well as greatest) fixed point.  
It is well known that the least fixed point (LFP) is $q^*$.\footnote{This
is not stated explicitly in \cite{Condon92}, who assumes for simplicity that the 
SSGs are stopping games, and thus whose equations have a unique fixed point;  but it follows easily from
well known facts.  See, e.g., \cite{EY-rssg-15}
for a generalization of this fact to a much richer class of (infinite-state) stochastic games.}

Consider now the ``$\beta$-discounted'' 
(or ``$\beta$-stopping'') version of these equations, 
$x=(1-\beta)F(x)$.
where each equation now has the form 
$x_i = (1-\beta) F_i(x)$,
for a given discount value $\beta \in (0,1)$.   We can also view these equations
as corresponding to a modified $\beta$-stopping version, $G^{\beta}$, of the original SSG, $G$,
where at each vertex there is a direct probability $\beta$ of immediately transitioning
to the $\mathbf{0}$-sink; and with the remaining probability, $(1-\beta)$, 
there remain exactly the same possibilities as before in $G$.)

Letting $F^\beta(x) := (1-\beta)F(x)$, note that $F^{\beta}:[0,1]^n \rightarrow [0,1]^n$
defines both a monotone map {\em and} a {\em contraction} map with respect
to the $l_\infty$ norm.
Specifically, for $x, y \in [0,1]^n$,    $\| F^{\beta}(x) - F^{\beta}(y) \|_\infty 
\leq (1-\beta) \| x - y \|_{\infty}$.
Hence, by Banach's fixed point theorem, $x=F^{\beta}(x)$ has a {\em unique} fixed point
solution, $q^\beta \in [0,1]^n$  (which is also both the least and greatest fixed point of
$x=F^{\beta}(x)$ in $[0,1]^n$).   The vector $q^\beta$ corresponds to the game values
of the $\beta$-stopping game $G^{\beta}$, starting at each vertex.

Let $|G|$  denote the bit encoding size of the given SSG, $G$.
There is a fixed polynomial, $h()$ such that for any SSG, $G$,
the denominator of the rational values 
$q^*_i$
is at most $2^{h(|G|)}$.
If we apply this to the already $\beta$-discounted SSG, $G^\beta$, then this says
that the denominators of the values $q^\beta_i$ are at most
$2^{h(|G| + \log(1/\beta))}$.

Moreover, for any SSG $G$, 
there is also a fixed polynomial, $r(x)$, such that given
a rational vector $q' \in [0,1]^n$,  such that $\| q^* - q' \|_{\infty} < 2^{-r(|G|)}$,
the closest rational number to $q'_i$ with denominator at most $2^{h(|G|)}$ is
$q^*_i$.

It is also known (see, e.g., Lemma 8 in \cite{Condon92})\footnote{Again, although
Condon's lemma is phrased assuming $G$ is a stopping game where edge probabilities
are always $1/2$, essentially the same proof with minor modification can be used to establish the analogous results in the
more general setting of non-stopping SSGs with arbitrary rational edge probabilities.}
that there is some fixed polynomial $t(\cdot)$, such that if $\beta  = \epsilon 2^{-t(|G|)}$,
for any $\epsilon \in (0,1)$,
then $\| q^* - q^\beta \|_{\infty} < \epsilon/2$.

Thus if we let $\epsilon = 2^{-r(|G|)}$,
and $\beta = \epsilon 2^{-t(|G|)}$, 
then not only do we have $\| q^* - q^\beta \|_\infty < 2^{-r(|G|)}$,
but we also have that, for all $i \in [n]$, the closest rational number to $q^\beta_i$
with denominator at most $2^{h(|G|)}$ is  $q^*_i$.

Next we note that for $\beta = 2^{-w(|G|)}$,  where $w(x) := r(x) + t(x)$ is a polynomial, the map $F^\beta: [0,1]^n \rightarrow [0,1]^n$
defines a  {\em polynomially contracting}  
function, as defined in
\cite{EY2010}, because for all $x, y \in [0,1]^n$,    $\| F^{\beta}(x) - F^{\beta}(y) \|_\infty 
< (1-\beta) \| x - y \|_{\infty}$.
In other words, the Lipschitz constant for the contraction map has the form $(1-\frac{1}{2^{poly(|I|)}})$,
where $|I|$ denotes the bit encoding size of the input $I$ that describes the map.
Hence, it follows from Proposition 2.2, part (3.) of \cite{EY2010}  that  in order
to compute some $q' \in [0,1]^n$  such that $\| q^\beta - q' \|_{\infty} < \epsilon/2$, 
for some desired $\epsilon \in (0,1)$, it suffices
to compute some $q' \in [0,1]^n$ such that $\| F^{\beta}(q') - q' \|_{\infty} < (\epsilon/2) \beta$.

Combining the above facts together it follows that, given an SSG, $G$, computing its vector of values $q^*$ is P-time 
reducible to computing a vector $q' \in [0,1]^n$  such that  $\|F^{\beta}(q') - q'| < \frac{1}{2^{z(|G|)}}$,
where $\beta = 2^{-w(|G|)}$, and where  $z(x) := w(x) +  r(x) = t(x) + 2 \cdot r(x)$, is a fixed polynomial. 

We next show that, given $G$, the problem of computing such a vector $q' \in [0,1]^n$ is  reducible to \Tarski{}.
Note, firstly, that for $\beta = 2^{-w(|G|)}$,
$F^{\beta}(x)$ 
is polynomial-time computable,
given the rational vector $x \in [0,1]^n$, and given the underlying SSG $G$.

We now define a discrete monotone function $H: [M]^n \rightarrow [M]^n$,
such that $H()$ is  a discretization of the monotone contraction map $F^{\beta}()$, 
where $\beta = 2^{-w(|G|)}$, such that 
any fixed point of $H$ directly yields (via rescaling) a vector
$q' \in [0,1]^n$ such that $\| F^{\beta}(q') - q' | < 2^{-z(|G|)}$.

The map $H()$ is defined as follows.
We let $M = 2^{2 \cdot z(|G|)}$.
For $v \in [M]^n$, we let $H(v)_i = \lfloor M \cdot F^{\beta}(v/M)_i \rfloor$, for all $i \in [n]$.
Clearly, $H:[M]^n \rightarrow [M]^n$ defines a monotone map
which is polynomial-time computable, given the input vector $v \in [M]^n$ and given the SSG, $G$.
Moreover,  if we find some fixed point $v^* \in [M]^n$
such that $v^* = H(v^*)$,  then $\| F^{\beta}(v^*/M) -  v^*/M \|_\infty < 2^{-z(|G|)}$.
Hence, a fixed point $v^*$ of $H$ immediately yields a vector $q' \in [0,1]^n$
such that $\|F^{\beta}(q') - q' \|_\infty < 2^{-z(|G|)}$, given  which we know we 
can compute $q^*$ in P-time.
We have therefore shown that the problem of computing the vector $q^* \in [0,1]^n$ of 
values for a given SSG, $G$, is P-time reducible to $\Tarski{}$.
\end{proof}

\subsection{Shapley's stochastic games reduce to \Tarski{}}

We now consider the original stochastic games
introduced by Shapley in \cite{Shapley53},
which are more general than Condon's games, and show
that approximating the value of such a game
(which is in general irrational, even when the input data associated 
with the game consists of rational numbers), to within
any given desired accuracy, $\epsilon >0$  (given in binary
as part of the input),
is polynomial time reducible to \Tarski{}.

Shapley's games are a class of 
two-player zero-sum ``stopping'', or equivalently
``discounted'', stochastic games. 
An instance of Shapley's stochastic game
is given by $G=(V,A,P,s)$,  where $V = \{v_1,\ldots,v_n\}$ is a set of 
$n$ vertices (or ``states''); $A = (A^1,A^2,\ldots,A^n)$
is a $n$-tuple of matrices, where,
for each vertex, $v_i \in V$,  $A^i$
is an associated $m_i \times n_i$ {\em reward matrix}, 
where $m_i$ and $n_i$ are positive integers 
denoting, respectively, the number of distinct ``actions'' available
to player 1 (the maximizer) and player 2 (the minimizer) at vertex $v_i$,
and where for each pair of such actions, $j \in [m_i]$ and $k \in [n_i]$, 
$A^i_{j,k} \in \rat$ is a reward for player 1
(which we assume, for computational purposes, is a rational number given as input by 
giving its numerator and denominator in binary).
Furthermore, for each vertex $v_i \in V$, and each
pair of actions $j \in [m_i]$ and $k \in [n_i]$,  $P^i_{j,k} \in [0,1]^{n}$ is
a vector of probabilities on the vertices $V$, such that
$0 \leq P^i_{j,k}(r)$, and $\sum_{r=1}^n P^i_{j,k}(r) < 1$,
i.e., the probabilities sum to {\em strictly less than} $1$.  
Again, we assume each such probability $P^i_{j,k}(r) \in \rat$ is a rational number
given as input in binary.   
Finally, the game specifies a designated start vertex $s \in V$.

A play of Shapley's game transpires as follows:  a token is initially
placed on $s$, the start node. Thereafter, during
each ``round'' of play,  if the token is currently on some node $v_i \in V$,
both players simultaneously and independently choose respective actions $j \in [m_i]$
and $k \in [n_i]$, and player 1 then receives the corresponding reward
$A^i_{j,k}$ from player 2; thereafter, for each $r \in [n]$ 
with probability $P^i_{j,k}(r)$ the token
is moved from node $v_i$ to node $v_r$, and
with the remaining positive probability $q^i_{j,k} = 1- \sum_{r=1}^n  P^i_{j,k}(r) > 0$, 
the game ``halts''.    Let $q = \min \{ q^i_{j,k} \mid i,j,k \} > 0$
be the minimum such halting probability at any state, and
under any pair of actions.    Since $q$ is positive,  i.e.,
since there is positive probability $\geq q > 0$ of halting
after each round,  a play of the game eventually halts with probability 1. 
The goal of player 1 (player 2) is to  maximize (minimize, respectively)
the expected total reward that player 1 receives from player 2 during
the entire play.   A {\em strategy} for each player specifies,
based in principle on the entire history of play thusfar,
a probability distribution on the actions available 
at the current token location.    Given strategies $\sigma_1$
and $\sigma_2$ for player 1 and 2, respectively, let $r_i(\sigma_1,\sigma_2)$
denote the expected total payoff to player 1, starting at node $s = v_i \in V$.
Shapley \cite{Shapley53} showed that these games have a {\em value}, meaning
that $\sup_{\sigma_1} \inf_{\sigma_2}  r_i(\sigma_1,\sigma_2)  = \inf_{\sigma_2} \sup_{\sigma_1} r_i(\sigma_1,\sigma_2)$.
In fact, Shapley showed that both players have optimal stationary 
(but randomized)
strategies in such games, i.e., optimal strategies that only depend on the
current node where the token is located, not the prior history of play, but where players
can randomize on their choice of actions at each node.

Let $r^*_i = \sup_{\sigma_1} \inf_{\sigma_2}  r_i(\sigma_1,\sigma_2)$ 
denote the game 
value starting at vertex $s = v_i \in V$.\footnote{
Note that we could also define $r^*_i$ as
$r^*_i = \max_{\sigma_1} \min_{\sigma_2} r_i(\sigma_1,\sigma_2)$, due 
to the existence of optimal strategies.}

\begin{proposition}
\label{prop:shapley-stoc-game-reduces-to-tarski}
The following total search problem is polynomial-time reducible to \Tarski{}:
Given an instance $G$ of Shapley's stochastic game, 
and given $\epsilon > 0$  (in binary), 
compute a vector $r' \in \rat^n$ such that $\| r^* - r' \|_{\infty} < \epsilon$.
\end{proposition}

\begin{proof}
For a matrix $B \in \real^{m_i \times n_i}$, let $\Val(B)$ denote
the von Neumann minimax value of the corresponding 2-player zero-sum (one-shot) matrix game defined by $B$.
Shapley showed that for an instance $G$ of his stochastic game, 
the $n$-vector $r^* \in \real^n$  of values starting at each vertex is 
the {\em unique} solution to the following system of $n$ equations in $n$ unknowns 
$x = (x_1,\ldots,x_n)$.
For each vertex $v_i \in V$, define the $m_i \times n_i$ matrix $B^i(x)$
whose $(j,k)$-entry  is 
$B^i(x)_{j,k} =  A^i_{j,k} + \sum_{r=1}^n P^{i}_{j,k}(r) x_r$.
The equations are given by:

\[   x_i =   \Val(B^i(x))   \quad   \mbox{for all $i \in [n]$}  \]

Denote this system of equations by $x=F(x)$.
If we let $M= max_{i,j,k} |A^i_{j,k}|$ denote the maximum absolute value reward, and let $W = \left\lceil \frac{M}{q} \right\rceil$,
then it is easy to check that  
$$F: \left[ -W,W \right]^n \rightarrow  \left[ -W,W \right]^n$$
defines a map from $\left[-W,W\right]^n$ to itself,
and moreover, as Shapley observed,  $F(x)$ is a {\em contraction} 
map with respect to the $l_\infty$ norm.  Specifically,
for any $x, y \in \left[-W, W \right]^n$,
$\| F(x) - F(y) \|_\infty \leq (1-q) \| x - y \|_{\infty}$.
In other words, the Lipschitz constant of the contraction map is
$(1-q)$.   (Hence, $F(x)$ is a {\em polynomially contracting}  
function, as defined in
\cite{EY2010}.)

Hence, by Banach's fixed point theorem,  $F(x)$ has a unique
fixed point in $\left[-W,W \right]^n$.    
Shapley showed that the unique fixed point is indeed 
the vector of game values $r^*$,
i.e., that $r^* = F(r^*)$ and $r^* \in \left[-W,W \right]^n$.
Furthermore, $F(x)$ is also clearly a {\em monotone} function,
even when the rewards $A^i_{j,k}$ can take on negative values.\footnote{
The monotonicity
of these maps was not explicitly noted by Shapley in \cite{Shapley53}, since his
proofs did not require the fact that the maps are monotonic, only that
they are contraction maps.}
This is because the rewards only play an additive role in the entries
$B^i(x)_{j,k} =  A^i_{j,k} + \sum_{r=1}^n P^{i}_{j,k}(r) x_r$
of the matrices $B^i(x)$,  and the coefficient $P^{i}_{i,k}(r)$ of each
variable $x_r$ is non-negative (it is a probability),  
and because the minimax value operator $\Val(\cdot)$
is a monotone operator.  In other words, for any two matrices $B, B' \in \real^{m_i \times n_i}$, 
if $B \leq B'$  (entry-wise inequality), then clearly $\Val(B) \leq \Val(B')$.

Thus $r^*$ is both the unique fixed point solutions of the polynomially contracting
map $F(x)$, as well as the (least and greatest) fixed point solution
of the monotone (Tarski) map $F(x)$.
Furthermore, $F(x)$ is a polynomial-time computable map:
given the input game $G$, and given a rational vector $b \in \rat^n$
(with entries encoded in binary),
we can compute $F(b)$ in time polynomial in the total bit encoding size
of $G$ and $b$.

Thus, just as in the case of the functions
$F^\beta(x)$ that arose for showing that computing the value of Condon's stochastic games
reduces to $\Tarski{}$, it follows 
from \cite{EY2010} (Proposition 2.2, part (3.)),
that in order to compute a vector 
$r' \in \left[-W,W \right]^n$,  
such that $\| r^* - r'\|_{\infty} < \epsilon$, it suffices to compute $r'$
such that 
$\|F(r') - r' \|_{\infty} < \epsilon \cdot q$.   

Hence, again as in the proof for Condon's game,
this allows us to ``discretize'' the monotone function $F()$,
to turn the problem of computing such a vector $r'$ into an
instance of $\Tarski{}$.
Specifically, for a positive integer $K$, let $\langle K \rangle = \{-K, -K+1, \ldots, 1, 0, 1, \ldots, K-1, K \}$.  Let $C = \left\lceil \frac{1}{\epsilon \cdot q} \right\rceil$.
We construct a discrete monotone map,  $H' : \langle M' \rangle^n \rightarrow \langle M'\rangle^n$,
where  $M' = W \cdot C$.
For $v \in \langle M' \rangle^n$,
we let $H'(v)_i := \lfloor C \cdot F( \frac{v}{C})_i \rfloor$, for all $i \in [n]$.
$H'$ defines a monotone map from $\langle M' \rangle^n$ to itself, which is polynomial-time computable,
given the input vector $v \in \langle M' \rangle^n$, given the
instance $G$ of Shapley's stochastic game, and given the 
desired error $\epsilon > 0$ (in binary).
Moreover (again, as in the case for Condon's game),  if we find
a fixed point $v^* \in \langle M' \rangle^n$ such 
that $v^* = H'(v^*)$,  then  $\| F(v^*/C) - v^*/C \|_{\infty} \leq \frac{1}{C} \leq \epsilon \cdot q$,
and hence $\| r^* - v^*/C \|_{\infty} < \epsilon$.

Hence we have shown that approximating the value vector $r^*$, for a given instance $G$
of Shapley's stochastic game,  within a given desired additive error $\epsilon > 0$
(given in binary),  is polynomial-time reducible to \Tarski{}.
\end{proof}

\section{Computing a Tarski fixed point on general finite lattices}

\label{sec:gen-lattice}

Thusfar we have studied the Tarski fixed point problem
only in the setting of monotone functions on the Euclidean grid
lattice $[N]^d$.   In this section we consider a
more general black-box model for monotone functions $f:L \rightarrow
L$, over an arbitrary finite lattice $(L,\preceq)$, where the lattice's
elements $L \subseteq \{0,1\}^n$ are encoded as binary strings of some
given length $n$, and where we assume the entire lattice $(L,\preceq)$
is known {\em explicitly} by the algorithm that queries the
function $f$.  The algorithm has unbounded
computational power, but only has oracle access to the monotone
function $f:L \rightarrow L$.      We call this framework
the {\em explicit-lattice black-box model}.

We shall show how any black-box query upper bound for computing a Tarski fixed point 
on the finite grid $[N]^d$  can be translated to a query upper bound in
this more general explicit-lattice black-box setting.
In particular, generalizing the $\log^d N$
algorithm for Euclidean grids,
we show that in this black box model there is a
deterministic algorithm that computes a fixed point of 
$f:L \rightarrow L$ using $O(\log^d(|L|))$ queries to the function $f$, where 
$d = dim(L,\preceq)$ is the {\em dimension} of
the lattice $(L,\preceq)$.  The {\em dimension}, $dim(L,\preceq)$, of a finite lattice
$(L,\preceq)$, and more generally the dimension of any partial order, 
can be defined as the smallest integer $d \geq 1$ such that the relation $\preceq$ is the
intersection of $d$ total orders on the same underlying set $L$  (\cite{DM41},\cite{Ore1962}, \cite{Trotter92}).  
Equivalently (see, e.g., \cite{Ore1962}, Theorem 10.4.2), it is the smallest $d \geq 1$ such that there
is an injective embedding  of $(L,\preceq)$ in the euclidean grid
$([|L|]^d,\leq)$, where $\leq$ is the standard coordinate-wise partial
order on $[|L|]^d$.

Note that a lower bound of
$\Omega(\log^2 (|L|))$ queries for computing a fixed point of a monotone
function $f:L \rightarrow L$ in this black box model follows directly from our
lower bound of $\Omega(\log^2 N)$ for monotone functions on the
2D grid $f:[N]^2 \rightarrow [N]^2$, just by letting $L = [N]^2$.
At present, we do not know any
better lower bound than $\Omega(\log^2 (|L|))$ in this black box model for 
arbitrary finite
lattices.   Let us note again that this black box model is very different
from the one considered in \cite{CLT08}, where the lattice itself
is not known explicitly, but 
is only accessible via an oracle for its partial order.
Hence the linear $\Omega(|L|)$ lower bound on the number of queries 
(including queries to the partial order itself) given 
in \cite{CLT08} for finding a fixed point
has no
bearing on the black box model described here, where the lattice itself is explicitly known,
and only the monotone function is given by an oracle.

We now show how query upper bounds for finding a Tarski fixed point
of a monotone function on grid lattices $[N]^d$ can be translated
to query upper bounds in  black box model for finding a Tarski 
fixed point of a monotone function $f:L \rightarrow L$ on
an arbitrary finite lattice of dimension $d$.

\begin{theorem}
\label{thm:gen-lattice-fixed-point}
Let $(L, \preceq)$ be a finite lattice of dimension $d = dim(L,\preceq)$,
and let $f: L \rightarrow L$ be a monotone function on  $(L, \preceq)$.

If for some function $g(N, d)$ there exists some black-box algorithm $A'$,
that computes a fixed point of any monotone function $f':[N]^d  \rightarrow
[N]^d$ in at most $g(N,d)$ queries to $f'$, then there exists an algorithm $A$ in the
explicit-lattice black-box model
that computes a fixed point of $f$ in at most $g(|L|,d)$ queries to $f$.
\end{theorem}

\begin{proof}
The lattice $L$ has dimension $d=dim(L)$.
By the mentioned equivalent characterization of the dimension (see \cite{Ore1962}, Theorem 10.4.2) 
there exists a injective map (an embedding) $\varphi: L \rightarrow [|L|]^d$ (which is certainly computable, since
$L$ is finite),  such that for all $a, b \in L$,  $a \preceq b$  if and only 
if $\varphi(a) \preceq \varphi(b)$.
The algorithm querying $f$ knows the lattice $L$ explicitly, and has unbounded
computational power, using which it first computes such an
embedding $\varphi:  L \rightarrow [|L|]^d$.

Let $N = |L|$.
Let  $Y = \varphi(L) \subseteq [N]^d$ denote the image of $L$ under
the embedding map $\varphi$.
We shall use $f$ to define a new monotone function $f': [N]^d \rightarrow [N]^d$  
(monotone with respect to the standard coordinate-wise partial order $\leq$ on
$[N]^d$), such that there is a (computable) $1$-to-$1$ correspondence between
the fixed points of $f$ and $f'$.
For  any $x \in [N]^d$,  let
$S(x) := \{ a \in L |  \varphi(a) \geq x \}$, denote the set of points in $L$ whose
$\varphi$-image in $[N]^d$ is ``above'' $x$.
For any subset $L' \subseteq L$, let $\bigwedge(L')$  denote
the greatest lower bound (meet) of $L'$ in the finite 
lattice $(L,\preceq)$.
For any $x \in [N]^d$,  let  $a_x :=  \bigwedge(S(x)) \in L$.
We define the function $f' :[N]^d \rightarrow [N]^d$
as follows: for all $x \in [N]^d$, $f'(x) :=  \varphi(f(a_x))$.
Note that for all $x \in [N]^d$,  $f'(x) \in Y$.
Note that for all $y \in Y$,   $a_y = \bigwedge(S(y)) = \varphi^{-1}(y)$.
Hence note that for any $b \in L$,  $a_{\varphi(b)} = \bigwedge(S(\varphi(b))) = b$.

Now, for $x, y \in [N]^d$,   suppose $x \leq y$. 
Then we must have  $S(x) \supseteq S(y)$,  and hence we have
$a_x \preceq a_y$, and thus 
by monotonicity of $f$ we have $f(a_x) \preceq  f(a_y)$.
Hence  $f'(x) = \varphi(f(a_x)) \leq  \varphi(f(a_y)) = f'(y)$.

Therefore $f': [N]^d \rightarrow [N]^d$ is a monotone map on $[N]^d$
whose range is a subset of $Y$, and hence whose fixed points are in $Y$.
Moreover, for any $x^* \in Y$, if $x^* 
= f'(x^*)$  is a fixed point, then $\varphi^{-1}(x^*) = \bigwedge(S(x^*))$ 
is a fixed  point of the original monotone map $f:L \rightarrow L$. 
Likewise, for any fixed point $b^* = f(b^*)$ of $f:L \rightarrow L$, we have 
$f'(\varphi(b^*)) = \varphi(f(a_{\varphi(b^*)})) = \varphi(f(b^*)) = 
\varphi(b^*)$, and hence $\varphi(b^*)$ is a fixed point of $f'$.

Now, if we have an algorithm $A'$ that computes a fixed point of 
$f':[N]^d \rightarrow [N]^d$ using at most $g(N,d)$ queries of $f'$,
we can use that algorithm  
in the explicit-lattice black box model to compute a fixed point of
$f:L \rightarrow L$ using the same number $g(N,d) = g(|L|,d)$  of queries to $f$.
Namely,  whenever algorithm $A'$ 
wishes to query
a point $x \in [N]^d$,   algorithm $A$ queries $f(a_x)$ and uses
that to compute $f'(x) = \varphi(f(a_x))$, returning it to $A'$.  
After at most $g(N,d)$ queries, algorithm $A'$ will output a fixed
point $x^* \in [N]^d$ of $f'$, and algorithm $A$ can then output the fixed point
$\varphi^{-1}(x^*) \in L$ of the function $f$, using the same number of queries to $f$.
\end{proof}

The following corollary follows immediately from Theorem \ref{thm:gen-lattice-fixed-point}, and the $O(\log^d(N))$-query recursive binary search algorithm
for computing a fixed point of a monotone function $f:[N]^d \rightarrow [N]^d$.
\begin{corollary}
There is an algorithm  in the explicit-lattice black-box model that, given access to
an oracle for a monotone function $f:L \rightarrow L$ over a finite lattice 
$(L,\preceq)$ 
with dimension $dim(L) = d$,  finds a fixed
point of $f$ using $O(\log^d(|L|))$ queries to $f$.
\end{corollary}

We remark that F\"{u}redi and Kahn \cite{FK88} 
have shown that for any $d \geq 1$ 
there exists a finite lattice $(L, \preceq)$ of dimension $d$
with size $|L| \in O(d^2 \log^2 d)$. 
In other words, the dimension of a finite lattice $(L, \preceq)$ can
be as large as $\Omega \left(\frac{\sqrt{|L|}}{\log(|L|)} \right)$.
For a general, expicitly given, finite partial order $(P,\preceq)$, the task 
of computing its dimension is \NP-hard,
and in fact already the task of deciding whether its dimension
is $\leq k$ for any {\em fixed} $k \geq 3$ is \NP-complete (\cite{Yan82}).   On the other hand, 
there is a polynomial time algorithm for computing the dimension of 
a (explicitly given) finite {\em distributive} lattice, $(L, \preceq)$,
based on a reduction to maximum matching. By Birkoff's representation theorem for finite distributive lattices
(\cite{Birk37}) and 
Dilworth's theorem (\cite{Dil50}) it can
be shown that the dimension of any distributive lattice is the
same as the {\em width} (i.e., maximum  anti-chain size)  of its set of join-irreducible
elements, and this width can be computed using maximum matching.  We forgo a more detailed description of this algorithm, because it is
not germain to the rest of this paper.
An intriguing tangential question that remains open as far as we know is:
what is the complexity of computing the dimension of a general, expilicitly given
(not necessarily distributive) finite lattice
$(L,\preceq)$?  The \NP-hardness proof in \cite{Yan82} uses partial orders
that are not lattices.

\section{Conclusions}

\label{sec:conclusions}

We have studied the complexity of computing a Tarski fixed point
for a monotone function over a finite discrete Euclidean grid, and 
we have shown that this problem essentially captures the complexity of computing a 
($\epsilon$-approximate) pure Nash equilibrium
of a supermodular game.   We have also shown that computing the value
of Condon's and Shapley's stochastic games reduces to this $\Tarski$ 
fixed point
problem, where the monotone function is given succinctly (by a boolean
circuit).

We have provided several upper bounds for the $\Tarski$ problem, showing
that it is contained in both $\PLS{}$ and $\PPAD$.
On the other hand, in the oracle model, 
for 2-dimensional monotone functions $f:[N]^2 \rightarrow [N]^2$, we have shown a 
$\Omega(\log^2 N)$
lower bound for the (expected) number of (randomized) queries required to find a Tarski fixed point,
which matches the $O(\log^d N)$ upper bound for $d=2$.
We have also shown that any upper bound, $g(N,d)$, on the
number of queries needed for finding
a fixed point of a monotone function $f:[N]^d \rightarrow [N]^d$
translates to an upper bound of $g(|L|,d)$ in the
explicit lattice black-box model for finding a fixed
point of any monotone function $f:L \rightarrow L$, on
any finite lattice $(L,\preceq)$ of dimension $d = dim(L)$.

\vspace*{0.18in}
\noindent {\em Subsequent work and open questions:} 
In the conference version of this paper, we had raised 
as an open question whether \Tarski{} is contained in the classes $\CLS$ (\cite{DP11})
and  $\EOPL{}$ (\cite{FGMS18,FGMS20}).
Since then, it has been shown that $\CLS = \PPAD \cap \PLS$  (\cite{FGHS23})
and that $\EOPL = \CLS$ \cite{GHJMPRT24}.
Hence  \Tarski{} is contained in $\CLS = \EOPL$.
It remains open whether  $\Tarski{}$  is \EOPL-hard,
or even just \UniqueEOPL-hard.  It also remains open whether
 $\Tarski{}$ is contained in $\UniqueEOPL$.

After the publication of the conference version
of this paper \cite{EPRY-itcs20},  in subsequent work Fearnley,
P\'{a}lv\"{o}lgyi, and Savani
\cite{FPS22} showed that for $d \geq 3$ there is a
$O(\log^{2\lceil \frac{d}{3} \rceil} N)$ query black-box algorithm for finding a Tarski fixed point
for $f:[N]^d \rightarrow [N]^d$. 
This in particular refuted an earlier conjecture of ours that
for a small fixed constant number of dimensions, say $d=3, 4$, 
a $\Omega(\log^d N)$ lower bound on the number of queries should hold (in particular, their upper bound yields a $O(\log^2 N)$ query algorithm when $d=3$).
Subsequently,  Chen and Li \cite{CL22} improved 
the upper bound further, providing a $O(\log^{\lceil \frac{d+1}{2} \rceil} N)$ query algorithm for finding a Tarski fixed point.
As an immediate corollary of Chen and Li's \cite{CL22} upper bound
and our Theorem \ref{thm:gen-lattice-fixed-point}, we obtain
the following:

\begin{corollary}
There is an algorithm  in the explicit-lattice black-box model that, given access to
an oracle for a monotone function $f:L \rightarrow L$ over a finite lattice 
$(L,\preceq)$ 
with dimension $dim(L) = d$,  finds a fixed
point of $f$ using $O(\log^{\lceil \frac{d+1}{2} \rceil}(|L|))$ queries to $f$.
\end{corollary}

On the lower bound side, Br$\hat{\mathrm a}$nzei et. al.  \cite{BPR25} recently showed a  $\Omega(\frac{d \cdot \log^2 N}{\log d})$ lower bound
for the expected number of queries required by any randomized
algorithm to find a fixed point of a monotone function $f:[N]^d \rightarrow [N]^d$,
using a suitable generalization of our family of herringbone functions in higher dimensions. Note that this bound subsumes our $\Omega(\log^2 N)$ lower bound: it is the same for fixed constant dimension $d$, and
is stronger when the dimension $d$ is large.

In another recent work, Chen, Li, and Yannakakis \cite{CLY23} addressed
a question we raised in the conference version \cite{EPRY-itcs20} of this paper regarding the
complexity of computing a Tarski fixed point when that fixed
point is (promised to be) unique.   They showed
that in the black-box oracle model, the worst-case number of 
queries required to find a fixed point for 
a monotone function $f:[N]^d \rightarrow [N]^d$, 
is the same (as a function of $N$
and $d$)  whether or not the function is promised to have
a unique fixed point or not.   This is surprising because
uniqueness of the fixed point for such a monotone function
imposes rather strong structural constraints on the behaviour of the function which do not hold in general.   It is worth noting however that the result
in \cite{CLY23} works only in the black-box model and does
not imply a polynomial-time reduction in the white-box model 
from the problem of computing a fixed point for
an arbitrary monotone function to that of computing
one for a monotone function with a (promised) unique
fixed point.  This is because their ``reduction'' maintains an exponential amount of information between queries,
and hence in its current form requires exponential time to implement.
In the same spirit, although it is not directly related to the
task of computing a Tarski fixed point,  a more recent work by Chen, Li, and Yannakakis \cite{CLY24,CLY25} has shown, surprisingly,  that in the black-box model finding an 
$\epsilon$-fixed point
for a  function $f:[0,1]^d \rightarrow [0,1]^d$
that defines a 
contraction map under the $l_\infty$-norm\footnote{Say, a $(1-\epsilon)$-contraction, 
but the specific contraction rate is not so crucial, see \cite{CLY24}.}
can
be done with $O(d \log (1/\epsilon))$ queries.
As observed in this paper, the task of computing the value
of Condon's and Shapley's stochastic games reduces
to the task of computing a Tarski fixed point for a monotone 
function.  Moreover, a closely related proof (see \cite{EY2010}) also shows that computing the value of Condon's and Shapley's stochastic
games reduces to computing an
$\epsilon$-fixed point for
a suitably succinctly presented contraction map $f:[0,1]^d \rightarrow [0,1]^d$
under the $l_\infty$ norm.  
Importantly, again, the black-box query algorithm 
of \cite{CLY24,CLY25} does
not at present yield a polynomial-time algorithm (as a function of 
$d$ and $\log(1/\epsilon)$) for computing such
a $\epsilon$-fixed point,  
where $f$ is assumed to be
succinctly presented, e.g., discretized and given by a boolean
circuit.  This is because their query algorithm again keeps track of an exponential amount of information in order to determine what point to query next, and the currently best known way to implement their algorithm requires exponential time.     An even more recent work
by Haslebacher et. al. \cite{HLSW25}, extending the 
results of \cite{CLY24,CLY25},  shows that 
 in the black-box model finding an 
$\epsilon$-fixed point
for a  function $f:[0,1]^d \rightarrow [0,1]^d$
that defines a 
$\lambda$-contraction map under {\em any} $l_p$-norm where $p \in [1,\infty) \cup \{\infty\}$, 
can
be done with $O(d^2 (\log (\frac{1}{\epsilon}) + \log (\frac{1}{1-\lambda})))$ queries.   However, again, 
for similar reasons as in the $l_\infty$ case of \cite{CLY24,CLY25}, these black box query upper bounds do not at present
yield a polynomial-time algorithm (as a function  of $d$, $\log(1/\epsilon)$, and $\log(1/(1-\lambda))$), for computing such a $\epsilon$-fixed point when the input function $f$ is assumed to be succinctly presented as input.

A recent work by Batziou et. al. (\cite{BFGMS25}) considers monotone $l_{\infty}$-contractions, i.e., functions $f:[0,1]^d \rightarrow [0,1]^d$  that are both monotone
and contractions with respect to the $\ell_\infty$ norm.
They provide a $O((c \cdot \log(1/\epsilon))^{\lceil d/3 \rceil})$-query algorithm, for some fixed constant $c$, for finding an $\epsilon$-fixed
point in this setting.  This is an (exponentially) worse
query upper bound than the $O(d \log(1/\epsilon))$-query
algorithm provided by \cite{CLY24,CLY25} for arbitrary contraction
functions.   However, an advantage of the result
of \cite{BFGMS25} for monotone contractions is
that each step of their query algorithm can be implemented
in time polynomial in the encoding size of the input (the succinct representation of the function $f$). 

A very recent work by Chen, Li, and Yannakakis (\cite{CLY25new})  provides a further improved
$O(\log^{\lceil \frac{d-1}{3} \rceil + 1} N)$-query black-box algorithm for finding a Tarski fixed point of a general monotone function $f:[N]^d \rightarrow [N]^d$.  In particular, this implies that $O(\log^2 N)$ queries suffice for finding a Tarski fixed point when $d=4$.
Their black-box algorithm, as it currently stands, does not automatically yield a white-box algorithm with the same time complexity
because it builds on the framework of \cite{CLY23} in which determining which query to make next requires, in principle, maintaining an exponential amount of
information between queries.

Even more recently, building on recent works of Chen, Li, and Yannakakis, 
a paper by Feodorov and Haslebacher (\cite{FeHas26}) has established a deterministic algorithm for computing an $\epsilon$-approximate fixed point of an $l_\infty$-contraction map using only $(\log \frac{1}{\epsilon})^{O(\sqrt{d} \log d)}$ queries {\em and} time.   Note that, in particular, this result subsumes and greatly strengthens the result established by \cite{BFGMS25} for
monotone $l_\infty$-contractions. 
The results of \cite{FeHas26} yield, in particular, the first deterministic subexponential-time algorithm for computing the value of  Condon's simple stochastic games, and the first (also deterministic) subexponential-time algorithm for $\epsilon$-approximating the value for Shapley's stochastic games.

  Another intriguing computational problem
  whose complexity status remains open, and which turns out to be
  P-time reducible to
  $\Tarski$, is the  $\Arrival$ problem,
  a certain kind of reachability problem on a directed graph,
  first defined and studied
  by Dohrau et. al. \cite{DGKMW17}.
  The input to the $\Arrival$ problem consists
  of a directed graph $G = (V,E,(\preceq_v)_{v \in V}, s, d)$ in which the out-going
  edges out of each
  vertex, $v \in V$, are totally orderered by the given $\preceq_v$. We are also given
  a start vertex $s \in V$,
  and a target vertex $d \in V$.  We consider the 
  deterministic walk that starts at $s$ and such that the first time 
  the walk visits any vertex $v \in V$ it immediately leaves via $v$'s
  first outgoing edge, and thereafter, if $v$ is visited again a second time,
  the walk leaves $v$ via its second outgoing edge, and so on, until all
  outgoing
  edges of $v$ have been used, after which the next time the walk visits $v$ it
  exits
  $v$ again using its first outgoing edge.
  So, in effect, this describes a deterministic walk along the directed edges of
  $G$ using a round-robin scheduler at each vertex in order to determine which
  outgoing edge to exit that node from next. 
  The $\Arrival$ problem asks, given such an input, whether starting from vertex $s$
  this deterministic walk will ever reach the target vertex $d$.

  Dohrau et. al. \cite{DGKMW17} showed that $\Arrival$ is in $\NP \cap \coNP$.
  Subsequent papers have shown
  that it is in $\UP \cap \coUP$  \cite{GHHKMS18},
  that a search version of $\Arrival$ is in $\PLS$ \cite{K17},
  and also in $\CLS (= \EOPL)$ \cite{GHHKMS18}, and even in $\UniqueEOPL$ \cite{FGMS20}.
  It remains open whether $\Arrival$ is in $\Ptime$-time.
  
  More recently, G\"{a}rtner et al. \cite{GHH21} used 
  the $O(\log^{O(d)}(N))$ algorithms for the Tarski fixed point problem
  to provide the first subexponential-time algorithm for
  $\Arrival$.  In fact, although they didn't
  mention this explicitly, it follows directly from their results that
  the $\Arrival$ problem (in both its search
  and decision version) is reducible in polynomial time to $\Tarski$.
  For completeness, we describe this reduction in more detail in
  Appendix \ref{appendix:B}.   Even more recently
  Haslebacher \cite{Has25} has shown that the $\Arrival$ problem is
  also reducible in polynomial time to the task of computing an
  approximate fixed point for a (succinctly presentable) $l_1$ contraction map $g:[0,1]^d \rightarrow [0,1]^d$ (the maps $g$ in Haslebacher's 
  reduction are in fact both monotone and $l_1$ contractions).

  Let us note that a more general version of the $\Arrival$ problem, called
  $\RecursiveArrival$, which considers the same kind of deterministic walk, but on the vertices
  of a succinctly represented infinite-state digraph (specifically, a 1-exit Recursive graph),
  has recently been shown to be contained in $\UniqueEOPL$ (\cite{Web23}), but is not
  yet known to be P-time reducible to $\Tarski$.

Let us mention that another important application domain in economics for Tarski's 
fixed point theorem is in the setting of stable matching
(\cite{GaleShapley1962})
and its many
generalizations (see, e.g., the books \cite{RothSotomayor1990,EcheniqueImmorlicaVazirani2023}).   It is well known (see, e.g., \cite{adachi2000})
that for any two-sided matching model with strict preferences the non-empty set of stable matchings can be characterized as the set of Tarski fixed points of a monotone
function $f: L \rightarrow L$ on a finite lattice $L$.   Indeed, if $M$ (``males'') and $W$ (``females'') denote
the set of agents on the two respective sides of the matching model, 
then the $(|M| + |W|)$-dimensional euclidean grid lattice $(L,\preceq)$ can be defined (\cite{adachi2000}) as the set $L = \times_{m \in M} (W \cup \{m\})   \times \times_{w \in W} (M \cup \{w\})$, where the partial order 
$\preceq$ on $L$
is defined as the coordinate-wise partial order, where in each
coordinate $m \in M$, we let $\preceq_m$ denote the (strict) preference total
order that agent $m \in M$ has over the set $W \cup \{ m \}$,\footnote{Agent $m$ ``prefers" itself over  $w \in W$
if it would rather remain unmatched than be matched with $w$.}
and for any $w \in W$, we let $\preceq_w$ denote the {\em reverse of} the preference total order that agent $w$ has over the set $M \cup \{w\}$.    
As is well-known,
the classic Gale-Shapley \cite{GaleShapley1962} ``male-proposal'' Deferred Acceptance (DA) algorithm 
computes (in polynomial time) a male-optimal stable matching, and this corresponds precisely 
to the greatest fixed point of this corresponding monotone
function (greatest from the perspective of males).
Indeed, standard Kleene iteration on these functions $f$ (as defined
in \cite{adachi2000}),
starting from the ``top'' element of the lattice $L$ will converge
to the greatest fixed point in at most $|M| + |W|$ iterations, and can be seen
to essentially mimic an accelerated version of the ``proposal'' rounds
of the DA algorithm.
Note that the lattice for the stable matching problem has dimension $d=|M|+|W|$ that is of the same order as (in fact larger than) the height $N = \max(|W|,|N|)+1$,
unlike the other applications (e.g. submodular games, stochastic games) we discussed in this paper,
where $d$ is much smaller (typically, exponentially smaller) than $N$.

There are many important models in the economics
literature which generalize the basic  Gale-Shapley two-sided stable matching model
(see, e.g., \cite{RothSotomayor1990,EcheniqueImmorlicaVazirani2023,EcheniqueOviedo2004,HatfieldMilgrom2005}) and for a number
of these models the existence of such a  ``generalized stable
matching'' is established via an application of Tarski's fixed point theorem. It will be interesting to investigate the complexity
of computing such generalizations of stable matchings, both in
the white-box and black-box model, in light of the connection
to  the \Tarski{} problem.

Clearly, a major remaining open question stemming from our work and the subsequent
improved upper bounds is: 
can we improve the current best (lower or upper) bounds on the number of queries required to find
a fixed point of a monotone function $f: [N]^d \rightarrow [N]^d$?
There is currently a very wide gap, indeed an exponential gap, between the best lower and
upper bounds in the black box model.
Of course, ultimately, we would like to know whether 
one can obtain a polynomial
time algorithm in the standard white box model for this task,
which would constitute a huge breakthrough on a number of fundamental computational problems. If no such algorithm can be found, is there a substantial complexity-theoretic impediment
to doing so (such as \EOPL- or \UniqueEOPL-hardness)?

\vspace*{0.08in}

\noindent {\large\bf Acknowledgements}

\vspace*{0.07in}

\noindent Thanks to Alexandros Hollender for pointing 
us to \cite{BussJohnson12}.  Thanks to Angus Joshi for comments
on an earlier draft.

\appendix

\section{Appendix A: \label{appendix:PPAD-closed-under-Turing-red}  proof 
that $\PPAD$ is closed under P-time Turing reductions.}

\label{appendix:A}

For the sake of completeness, we provide a proof of a
result stated by Buss and Johnson
\cite{BussJohnson12} (see also \cite{Johnson2011}), but with the proof left to the reader there,
namely that $\PPAD$ is closed under polynomial-time Turing reductions.

\begin{theorem}[cf. \cite{BussJohnson12}, \cite{Johnson2011}]
$\FP^{\PPAD} = \PPAD$.
\end{theorem}
\begin{proof}
Suppose that a total search problem $Q$ is in $\FP^{\PPAD}$. In  other words, there is a P-time 
Turing machine $M$,
that solves $Q$ using an oracle from $\PPAD$, which wlog we can
assume is an oracle for the $\PPAD$-complete problem End-Of-Line (EOL) 
(\cite{Pap94}).
Recall that EOL is a total search problem defined as follows:
a valid input $I$ consists of two boolean circuits $S$ (``successor'') and $P$
(``predecessor''), each with $m$ input bits and $m$ output bits,  where  moreover $P(0^m)$ = $0^m$, $S(0^m) \neq 0^m$,  and $P(S(0^m)) = 0^m$. 
The circuits succinctly describe a directed graph $G_I = (V_I, E_I)$,
where $V_I = \{0,1\}^m$ 
and $E_I = \{  (u,v) \in V_I \times V_I \mid  u \neq v, S(u) = v, \mbox{and} \ P(v) = u \}$.  
Note that all nodes of $G_I$ have in-degree $\leq 1$ and out-degree $\leq 1$.
Note furthermore that the default source node $0^m \in V_I$ has in-degree $=0$ and out-degree $=1$, and hence it has sum total degree  $=1$.  We call any node whose sum total degree is $= 1$ an {\em endpoint}.
The task is to output a different endpoint node, $v^*  \in V_I$, $v^* \neq 0^m$,
whose sum total degree is also equal to $1$  (i.e., where $v^*$ is either a ``source'' 
node other than $0^m$, or a ``sink'' node).   Such a node $v^*$ 
must exist by a simple counting argument:
the total number of outgoing edges, summed over all nodes of $G_I$, must equal the total number of
incoming edges. But the default node $0^m$ has one outgoing  and no incoming edge, 
thus there exists a node with more incoming than outgoing edges, and that node must have in-degree $=1$
and out-degree $=0$.\footnote{Indeed, this argument shows a sink node must 
exist, but the 
EOL problem does not insist on outputing a sink.  The defining complete problem for another  total search complexity class, \PPADS{}, is the problem that
asks to output a sink node given the same input as the EOL problem.}

Given an input $x$ of length $n$, the machine $M$ generates adaptively a sequence
of at most $p(n)$ queries to EOL, for some polynomial $p(\cdot)$, where each query may depend on the 
previous answers,
and at the end outputs a solution to $Q$ for the input $x$.
Each query specifies a succinct graph for EOL (via polynomial-sized circuits for the predecessor 
and successor
of each node). We can assume wlog that the nodes of all these graphs are
binary strings of a certain length $m$ (polynomial in $n$), that 
the default source node in each graph is $0^m$ (this is only for simplicity, and isn't necessary to assume), and that furthermore each solution to the problem $Q$ is a string of length $m$.
We will define one ``big'' graph $G$ for EOL (i.e., predecessor and successor circuits $P$ 
and $S$)
such that every non-default endpoint of $G$ provides a solution of $Q$ for $x$.

The nodes of $G$ are of the form $(\langle y_1, \ldots , y_k \rangle, v)$, for some $k$, where $0 \leq k \leq p(n)$,
where the $y_i$ and $v$ are binary strings of length $m$.
Such a tuple can be represented of course by a string of length $O(p(n)m)$,
for example by the string $1 y_1 1 y_2 \ldots 1 y_k 0 \ldots 0 v$ (with enough $0$'s so all the  strings have exactly
the same length, say $(m+1)p(n) +m$).
Call such a tuple $(\langle y_1, \ldots, y_k\rangle, v)$ 
{\em valid} if $y_1$ is a legal answer to the 
first query generated by $M$
on input $x$, $y_2$ is a legal answer to the second query generated (which 
depends on $x$ and $y_1$),
$y_3$ is a legal answer to the third query, and so on, and moreover either $(i)$ 
after the 
sequence $y_1, \ldots, y_k$ 
of answers, machine $M$ does not generate any further queries but returns 
solution $v$,
or $(ii)$ $M$ generates a $(k+1)$-th query and $v$ is any node of the 
corresponding graph
(i.e. a string of length $m$, since each such corresponding graph has all strings of length $m$ as 
its nodes).

The nodes of the graph $G$ are all valid tuples. Or in other words, we can 
consider as
nodes all strings of length $(m+1)p(n) +m$, and those strings $w$ that do 
not correspond
to valid tuples have self-loops , i.e., we let $S(w)=P(w)=w$ (hence they 
are not solutions).
Clearly, given a string $w$, we can tell in polynomial time if it 
corresponds to a
valid tuple.
The default source node is the tuple $(\langle \rangle, 0^m)$,
where ``$\langle \rangle$'' denotes an empty tuple. 
It remains to define $S$ and $P$ for valid tuples.

Let $(\langle y_1, \ldots, y_k \rangle, v)$ be a valid tuple where there is a $(k+1)$-th query.
Every $y_i$ is a sink or a source in its corresponding query graph.
Suppose that the number of $y_i$ that are sources is ${\mathbf{even}}$.
If $v$ is not a sink in the $(k+1)$-th graph then $S((\langle y_1, \ldots, y_k \rangle, v)) = 
(\langle y_1, \ldots , y_k \rangle, s(v))$
and if $v$ is not a source then $P((\langle y_1, \ldots, y_k \rangle, v)) = (\langle y_1, \ldots, y_k \rangle, p(v))$,
where $s$ and $p$ are the successor and predecessor functions in the $(k+1)$-th 
graph.
If $v$ is a sink then if $M$ makes a $(k+2)$-nd query after seeing the
answer $v$ to its $(k+1)$-st query, then we let $S((\langle y_1, \ldots, y_k \rangle, v)) = (\langle 
y_1, \ldots , y_k,v \rangle,0^m )$;
otherwise $M$ returns a solution $u$ and we let $S((\langle y_1, \ldots, y_k\rangle, v))
=(\langle y_1, \ldots, 
y_k,v \rangle, u)$,
and $P((\langle y_1, \ldots, y_k,v \rangle, u)) = (\langle y_1, \ldots, y_k \rangle, v)$, and 
$(\langle y_1, \ldots , y_k,v \rangle, 
u)$ is a sink of the graph $G$  (from which the solution $u$ to $Q$ is easily extractable).
If $v$ is a source $\neq 0^m$ then if $M$ makes 
a $(k+2)$-nd query after seeing the answer $v$ to its 
$(k+1)$-st query, then we let $P((\langle y_1, \ldots, y_k \rangle, v)) = 
(\langle y_1, \ldots, y_k,v \rangle,0^m 
)$;
otherwise $M$ returns a solution $u$ and we let $P((\langle y_1, \ldots, y_k\rangle, v))=(\langle y_1, \ldots, 
y_k,v \rangle, u)$,
and $S((\langle y_1, \ldots, y_k,v \rangle, u)) = (\langle y_1, \ldots, y_k \rangle, v)$ and $(\langle y_1, \ldots, y_k,v \rangle, 
u)$ is a source of the graph $G$ (from which, again, the solution $u$ to $Q$ can be easily
extracted).
Lastly, if $v=0^m$ then $P((\langle y_1, \ldots, y_k \rangle, 0^m) ) = ( \langle y_1, \ldots,
y_{k-1} \rangle, y_k)$.

If, on the other hand, 
the number of $y_i$ that are sources is ${\mathbf{odd}}$, then we reverse the 
directions of the edges, as follows.
If $v$ is not a source in the $(k+1)$-th graph then 
$S((\langle y_1, \ldots, y_k \rangle, v)) = 
(\langle y_1, \ldots, y_k \rangle, p(v))$
and if $v$ is not a sink then $P((\langle y_1, \ldots, y_k \rangle, v)) = 
(\langle y_1, \ldots, y_k \rangle, s(v))$,
where $s$ and $p$ are the successor and predecessor functions in the $(k+1)$-th 
graph.
If $v$ is a sink then if $M$ makes a $(k+2)$-nd query after seeing the answer $v$
to its $(k+1)$-st query, then we let
$P((\langle y_1, \ldots, y_k \rangle, v)) = (\langle y_1, ..., y_k,v \rangle,0^m )$;
otherwise $M$ returns a solution $u$ and we let $P((\langle y_1, \ldots, y_k \rangle, v))=
(\langle y_1, \ldots, 
y_k,v \rangle, u)$,
and we let $S((\langle y_1, \ldots, y_k,v \rangle, u)) = (\langle y_1, \ldots, y_k \rangle, v)$, and $(\langle y_1, \ldots, y_k,v \rangle, 
u)$ is a source of the graph $G$ (from which, again, the solution $u$ to $Q$ can
easily be extracted).
If $v$ is a source $\neq 0^m$ then if $M$ makes a $(k+2)$-nd query
after seeing the answer $v$ to its $(k+1)$-query, then we let 
$S((\langle y_1, \ldots, y_k \rangle, v)) = (\langle y_1, \ldots, y_k,v \rangle,0^m )$;
otherwise $M$ returns a solution $u$ and we let $S((\langle y_1, \ldots, y_k \rangle, v))= (\langle y_1, \ldots, 
y_k,v \rangle, u)$,
and $P((\langle y_1, \ldots, y_k,v \rangle, u)) = 
(\langle y_1, \ldots , y_k \rangle, v)$ and $(\langle y_1, \ldots, y_k,v \rangle, 
u)$ is a sink of the graph $G$ (from which, again, the solution $u$ to $Q$ can
easily be extracted).
If $v=0^m$ then $S((\langle y_1, \ldots, y_k \rangle, 0^m) ) = ( \langle y_1, \ldots,
y_{k-1} \rangle, y_k)$.

This completes the definition for the ``big'' EOL instance
$G$, from any of whose non-default endpoints we can easily extract a 
solution for $Q$.
The key point is that separating the $\mathbf{even}$ and $\mathbf{odd}$ case,
and reversing the direction of the edges in the case where the number of $y_i$'s
that are sources is $\mathbf{odd}$,  enables us to efficiently ``attach'' the corresponding directed line graphs for different EOL queries carried out by $M$ to each other, to form one ``big''  directed line
graph $G$, without having to directly ``compute'' solutions for any of the
smaller EOL instances.

This completes the proof that any total search problem $Q$ in $\FP^{\PPAD}$
is already in $\PPAD$.   
\end{proof}

We note that an easy adaptation of the same proof also shows that $\FP^\PPA = \PPA$, since
the canonical complete problem for \PPA{} is just an undirected version of the EOL problem.
Moreover, essentially the same proof also shows that $\FP^{
\PPADS} = \PPADS$. Indeed, the case of $\PPADS$ is easier than $\PPAD$ and doesn't require
splitting into $\mathbf{even}$ and $\mathbf{odd}$ cases (and reversing direction of edges
in odd cases), because when using a $\PPADS$ oracle that always returns a sink of the given EOL input instances,
all the $y_i$'s in valid tuples in the above reduction will be sink nodes, and hence there are 
no $\mathbf{odd}$ valid tuples.

\section{Appendix B:   $\Arrival$ is P-time reducible to $\Tarski$}

\label{appendix:B}

In this section we elaborate on the polynomial time reduction,
mentioned in the Conclusions section, 
from the $\Arrival$ problem to $\Tarski$, which follows from
the results of \cite{GHH21}, who provided a subexponential-time
algorithm for $\Arrival$ by exploiting 
available algorithms for the $\Tarski$ problem.

  We can view an instance of the $\Arrival$ problem (\cite{DGKMW17,GHH21}) as being given by
  $\mathcal{A} = (V, g^0, g^1, s, d, d')$,
  where $V$ is a finite set of vertices,
  $s \in V$ is a distinguished
  start vertex, and $d,d' \not\in V$ are
  two  distinguished terminal vertices
  (they are not included in $V$ for notational convenience later), 
  and the functions $g^i: V \rightarrow V \cup \{d, d'\}$,
  $i \in \{0,1\}$,  define two outgoing edges
  $(v,g^0(v))$ and $(v,g^1(v))$  out of each vertex $v \in V$
  to any vertex in $V \cup \{d, d'\}$.\footnote{The restriction to only two outgoing edges from each vertex in $V$
  is w.l.o.g., and
  doesn't make the problem any easier.   This is clearly implicit
  already in \cite{DGKMW17}. To see why, note that
  if a vertex $v \in V$ has $k > 2$ outgoing edges, 
  those edges can be replaced by introducing an auxiliary directed complete binary tree, $T_v$, 
  of depth $D = \lceil \log_2(k+1) \rceil$, rooted at $v$.
  The leaves of $T_v$ are ordered
  lexicographically according to the {\em reverse} of
  the binary string labeling the path from the root $v$ of $T_v$ 
  to each leaf.
  Let these leaves, ordered accordingly, be $v'_0,v'_1, \ldots , v'_{2^D-1}$.
  The leaf, $v'_i$,  $0 \leq i \leq k-1$, is
  identified with the head node of the $i$'th  outgoing edge from the
  original node $v \in V$. 
  For each of the remaining $D-k$ leaves, $v'_j$, $j \geq k$, of $T_v$,   we let $g^0(v'_j) := g^1(v'_j) := v$.  (And then we can 
  merge nodes, as described in the next footnote, to get rid of all such nodes 
  $v'_j$, $j \geq k$, both of whose outgoing edges lead to the same node in the new digraph.)
  It is simple to check that this construction ensures that the switching walk
  on the new outdegree-2 directed graph, obtained by augmenting 
  the original instance $\mathcal{A}$ with these auxiliary binary trees, visits exactly the same sequence
  of non-auxiliary nodes as the switching walk on $G$  (with a bounded
  number of visits to auxiliary nodes in between any two such visits
  to non-auxiliary nodes).}
  The role of the two terminal vertices $d$ and $d'$, which have no outgoing
  edges, will be made clear shortly ($d$ is
  the ``good'' target vertex 
  and $d'$ is the ``bad'' dead end). 

Associated with an $\Arrival$ instance,
$\mathcal{A}$, is a directed graph called the {\em switch graph}, $G(\mathcal{A}) = (V \cup \{ \hat{s},d,d' \}, E)$,
where $\hat{s} \not\in V \cup \{d,d'\}$ is an extra artificial start vertex
(included for notational convenience), 
and where the edge set of $G(\mathcal{A})$ is
  $E = \{ (\hat{s}, s) \} \cup \{ (v, g^0(v)) \mid v \in V   \}  \cup \{ 
  (v, g^1(v)) \mid
  v \in V \}$.   So, $\hat{s}$ has no incoming
  edges, and the only outgoing edge from $\hat{s}$
  is $(\hat{s},s) \in E$.

  Furthermore, we assume wlog
  that $\mathcal{A}$ has the following properties
  which we can ensure via some fairly simple polynomial time 
  transformations,
  which don't change its status as a Yes or No instance of $\Arrival$.
  For any vertex $v \in V \cup \{ \hat{s} \}$, there exist 
  directed paths in $G(\mathcal{A})$ from $v$ to both of the terminal vertices
  $d$ and $d'$.
  Furthermore, for all vertices $v \in V$, $g^0(v) \neq g^1(v)$.\footnote{We can assume this because any vertex $v \in V$ with $g^{0}(v) = g^1(v)$
  can be eliminated from the
  switch graph $G$.  If some other vertex $u \in V$ 
  has $g^i(u) = v$, $i \in \{0,1\}$, then we can 
  let $g^i(u) := g^0(v) (= g^1(v))$.
  Repeatedly carrying out this step until no further vertices
  can be removed ensures no remaining vertices where $g^0(v) = g^1(v)$.}
  An $\Arrival$ instance that satisfies all the above properties
  is called a {\em terminating} $\Arrival$ instance.

  A {\em switching run} on $\mathcal{A}$, starting
  at a vertex $v_0 \in V$,
  proceeds deterministically in discrete time steps, as follows.
  For each time $t = 0,1,2,\ldots$, we use a function
  $h_t: V \rightarrow \{0,1\}$ to denote the current {\em switching-state}.
  The run 
  consists of a sequence  $\langle (v_t,h_t) \rangle_{0 \leq t \leq T}$
  of vertex and switching-state pairs, $(v_t,h_t)$, for $0 \leq t \leq T$.
  where the termination time $T$ will be defined later (it can be at most exponential in $|V|$).  
  The initial switching state, $h_0$, is defined
  as $h_0(v) := 0$ for all $v \in V$.
  Inductively, if the prefix of the run up until time $t$ consists
  of $(v_0,h_0), (v_1,h_1), \ldots, (v_t,h_t)$, then,
  firstly, if $v_t \in \{d,d'\}$ then the run is said to {\em terminate}
  immediately at vertex $v_t$ at time step $t$,
  and otherwise if $v_t \in V$,
  the next pair 
  $(v_{t+1}, h_{t+1})$ at time step $t+1$ is defined as follows:
  $v_{t+1} := g^{h_t(v_t)}(v_t)$ \ ; \   $h_{t+1}(v_t) := 1 - h_t(v_t)$, \ and
  for all $v \in V \setminus \{v_t\}$,
  \ $h_{t+1}(v) := h_t(v)$.

  The $\Arrival$ decision problem is defined as follows:
  given such an instance, $\mathcal{A}$, starting at the designated
  start vertex $v_0 := s$, does the run eventually terminate at vertex $d$,
  i.e., does there exist a time $t' \in \nat$ such
  that $v_{t'} = d$?

  As mentioned,
  G\"{a}rtner et al. \cite{GHH21} used
  the available $O(\log^{O(d)}(N))$ algorithms for the Tarski fixed point problem
  to provide the first subexponential time algorithm for
  the $\Arrival$ problem.  In fact, although they didn't
  mention this explicitly, their results readily imply 
  that the $\Arrival$ problem (in both its search
  and decision version) is reducible in P-time to $\Tarski$,
  as we now explain.


It was shown in (\cite{DGKMW17}, Lemmas 2 and 3),
  that for any terminating $\Arrival$ instance, $\mathcal{A}$,
  starting at any vertex $v_0 \in V$
  the switching run will eventually reach
  (and hence ``{\em terminate}'' at) one of the two
  terminal vertices, $d$ or $d'$, and that it will
  do so after crossing each edge at most $2^n-1$ times, where $n = |V|$.
Consider the following {\em switching-run-edge-count (SREC)} (or ``edge-profile'') 
function
  $x^*: E \rightarrow [2^n]$,
  defined by letting $x^*(e)$ be the number of times that directed edge $e \in E$
  is traversed in the switching run starting at vertex $s$, with the stipulation by definition
  that the artificial edge $(\hat{s},s)$ is traversed
  exactly once at the ``start'' of the run, and hence
  that $x^*((\bar{s},s)) = 1$.\footnote{The artificial start state $\hat{s}$ is used
  for notational convenience later.}
This SREC function satisfies some basic constraints that we will
now describe.
 Dohrau et. al. \cite{DGKMW17} defined the key notion of
  a {\em switching edge-flow (SEF)}, which is any function $x: E \rightarrow \nat$
 which satisfies the following constraints
  capturing ``flow conservation'' and
  ``switch consistency''.  To facilitate defining the constraints,
   for each $v \in V$,
  let $E^+(v) = \{ (v,w) \in E \mid  w \in V \cup \{d,d'\} \}$ denote the set of outgoing
   edges from $v$, and for each $v \in V \cup \{d,d'\}$ let
    $E^-(v) = \{ (w,v) \in E \mid w \in V \cup \{ \hat{s} \} \}$ denote the set of incoming edges to $v$ (including possible the edge from $\hat{s}$ in case $v=s$).
  The constraints defining a SEF are:
\begin{eqnarray}
 x((\hat{s},s))  = 1   \ \  ;
 &&   \mbox{(the flow of $1$ ``starts'' on $(\hat{s},s) \in E$)}   \label{eq1-sef} \\
    \sum_{e \in E^+(v)} x(e) - \sum_{e \in E^-(v)} x(e)  =  0   \ , &&  \mbox{$\forall v \in V$; \ \ (``flow conservation'')}
    \label{eq2-sef} \\
    x((v,g^1(v))) \  \leq \  x((v,g^0(v))) \  \leq \ x((v,g^1(v))) + 1 \ ,  && \mbox{$\forall v \in V$;   \ (``switch consistency'')}  \label{eq3-sef}
  \end{eqnarray}

Clearly, the SREC function $x^*$ itself satisfies 
all these constraints, and hence constitutes a switching edge-flow (SEF).
However, it is not necessarily unique: there can be other switching edge-flows (\cite{DGKMW17}).  Nevertheless, Dohrau et. al. (\cite{DGKMW17}, Lemma 1)
showed that the SREC function $x^*$ constitutes the {\em least} SEF, in the sense that for any  
SEF, $x':E \rightarrow \nat$, we must have $x^* \leq x'$,
meaning $x^*(e) \leq x'(e)$, for all $e \in E$.
Note that therefore $x^*$ also minimizes the sum $\sum_e x(e)$
among all SEFs.

Note that for any SEF, $x : E \rightarrow \nat$, 
the flow conservation constraints imply that
exactly one of the following two cases holds:
(A) $\sum_{e \in E^-(d)} x(e) = 1$
  and $\sum_{e \in E^{-}(d')} x(e) = 0$, or else
  (B) $\sum_{e \in E^-(d')} x(e) = 1$ and $\sum_{e \in E^-(d)} x(e) = 0$.   However, since the SREC function
  $x^*$ satisfies one of
  these two and  is also the least SEF, this implies that for any given terminating $\Arrival$
  instance $\mathcal{A}$,
   either {\em all} SEFs of $\mathcal{A}$ satisfy (A),
  or {\em all} SEFs of $\mathcal{A}$ satisfy (B).

Thus, in order to decide the $\Arrival$ problem
all we need to do (as already observed in \cite{DGKMW17}) is to compute a
(any) SEF, $x : E \rightarrow \nat$,
for the given terminating
$\Arrival$ instance, and then check whether $x$
satisfies $\sum_{e \in E^-(d)} x(e) = 1$.

Given any SEF, $x: E \rightarrow \nat$, consider the
corresponding
function on vertices $y^{[x]}: V \rightarrow \nat$, defined
by letting 
\begin{equation}
y^{[x]}(v) := \sum_{e \in E^{-}(v)} x(e) \quad , \quad \mbox{for all
$v \in V$.} 
\label{eq:def-y-up-x}
\end{equation}

Note, in particular, that if $x^*$ is the SREC function,
then $y^{[x^*]}: V 
\rightarrow \nat$  is equivalent to the {\em switching run vertex-count (SRVC)}
function, denoted $y^*$, which by definition 
counts the number of times, $y^*(v)$, that
each vertex, $v \in V$, is visited on the switching run.
It was
shown in (\cite{GHH21}, Lemma 1) that the number
of times each vertex in $V$ is visited on the switching
run is at most $2^n$, and hence
we can write $y^*= y^{[x^*]}$  as $y^*: V \rightarrow [2^n]$.



For $v \in V$  and $j \in \{0,1\}$, let
$H^{-}_j(v) = \{ u \in V \mid g^j(u) = v \}$.
For any function $y: V  \rightarrow \nat$,
consider the following system of constraints:

\begin{eqnarray}
y(v)  & = &   \sum_{ u \in H^-_0(v)} \left\lceil \frac{y(u)}{2} \right\rceil   + \sum_{u \in H^-_1(v)} \left\lfloor \frac{y(u)}{2}
\right\rfloor +  \delta_{=s}(v) \quad,  \quad  \label{eq1-svf}  
\mbox{for all $v \in V$ } \ ,  
\end{eqnarray}
where, by definition, $\delta_{=s}(v) = 0$ if $v\neq s$ and $\delta_{=s}(v) = 1$ if $v=s$.
We call any function $y: V  \rightarrow \nat$ that satisfies these constraints a {\em switching vertex-flow (SVF)}.
For any SVF, $y: V \rightarrow \nat$,
let $x^{[y]}: E \rightarrow \nat$ denote the function obtained
by letting, for all $v \in V$:
\begin{equation}
x^{[y]}((v,g^0(v)) := \left\lceil \frac{y(v)}{2} \right\rceil
\quad \mbox{and} \quad x^{[y]}((v,g^1(v)) := \left\lfloor \frac{y(v)}{2} \right\rfloor  \quad , \ \  \mbox{and} \quad x^{[y]}((\hat{s},s)) := 1.
\label{eq-x-to-y}
\end{equation}

The following claim can be shown easily.

\begin{claim}\mbox{}

\begin{enumerate}
\item For any SEF, $x$: \   $y^{[x]}$ is a SVF. \ For any SVF, $y$: \   $x^{[y]}$ is a SEF.
\item For any SEF, $x$: \  $x^{[y^{[x]}]} = x$. \  For any SVF, $y$: \ 
$y^{[x^{[y]}]} = y$.
\item For any two SEFs, $x$ and $x'$,  if $x \leq x'$
then $y^{[x]} \leq y^{[x']}$.\\
For any two SVFs, $y$ and $y'$, if $y \leq y'$
then $x^{[y]} \leq x^{[y']}$.
\end{enumerate}
\end{claim}

Hence, there is an easy-to-compute one-to-one correspondence between
SEFs and SVFs, which respects 
the coordinate-wise partial order between SEFs and SVFs 
(i.e., gives an isomorphism between the respective partial orders). 
Moreover, as noted, the least element of the SEFs is
the SREC, $x^*$, whereas the least element of
the SVFs  is the SRVC, $y^* = y^{[x^*]}$.


Hence, in order to decide the $\Arrival$ problem, it also 
suffices to compute a (any)  SVF.

G\"{a}rtner et al \cite{GHH21} used a Tarski fixed point 
characterization of a ``multi-run''
generalization of SVF functions, corresponding to
a ``multi-run'' generalization of SEFs. (They called their 
generalization of a SEF 
a ``candidate switching flow'').
For our purposes, we will not need to describe what such
multi-run generalizations are.
Although they did not mention this
explicitly, a special case of
their results, namely the case where we let the set
$S \subseteq V$ of their multi-run problem be all vertices, i.e, $S := V$, provides a
Tarski fixed point characterization of SVFs as follows.

Consider the defining system of equations (\ref{eq1-svf})
for an SVF, $y$.  This can be viewed as a system of
$n = |V|$ equations, 
in $n$ variables,  $\langle y_v \rangle_{v \in V}$, 
with one equation $y_v = F_v(y)$ for each vertex
$v \in V$.   We can write this entire
system in vector notation as $y=F(y)$.
Note, importantly, that $F: \nat^{n} \rightarrow \nat^{n}$ 
defines a monotone function on $\nat^{n}$, because all
coefficients on the RHS of equations (\ref{eq1-svf}) are non-negative.
Alas,  $\nat^{n}$ is not a finite lattice.
G\"{a}rtner et al (\cite{GHH21}, Lemma 9) considered
the following variant of $F(y)$, which imposes a finite ``top''.
For $n = |V|$, let $N= 2^{n}$, and for each $v \in V$, let 
\begin{equation}
      D_v(y) :=   \min ( F_v(y),  N) 
\end{equation}
Let $[N] = \{0,\ldots,N\}$.  
Clearly $D:[N]^{n} \rightarrow [N]^{n}$ defines 
a monotone function from the finite lattice $[N]^{n}$ to 
itself, and hence has a fixed point.

\begin{proposition}[\cite{GHH21}, Lemma 9,
    when letting $S:= V$]
Every fixed point solution $y' \in [N]^{n}$ of $y = D(y)$ 
is also a fixed point solution of $y = F(y)$.
\end{proposition}

%

Hence, it follows that, given an arrival instance $\mathcal{A}$, 
in order to compute an SVF for it, and hence to decide the 
$\Arrival$ problem, it suffices
to compute a (any) fixed point of the discrete monotone function 
$D:[N]^{n} \rightarrow [N]^{n}$, whose fixed point equations $y = D(y)$ can easily
be described (and $D(y)$ can be evaluated) in P-time given the instance $\mathcal{A}$ (and given, additionally $y \in [N]^{n}$).
This shows
\begin{proposition}
    $\Arrival$ is polynomial time reducible to
$\Tarski$.
\end{proposition}

\end{document}